\documentclass[11pt]{article}

\usepackage{a4wide}
\usepackage{color}
\usepackage{amsfonts,amsmath,amsthm}

\newcommand{\eps}{\varepsilon}

\newcommand{\prob}{\mathbb{P}}

\newcommand{\E}{{\mathcal E}}

\newcommand{\C}{{\textsc C}}

\newcommand{\Bin}{\mathrm{Bin}}

\newcommand{\N}{{\mathcal N}}

\newcommand{\Suc}{{\mathrm{SUC}}}
\newcommand{\Tr}{{\mathrm{TRUE}}}
\newcommand{\Fa}{{\mathrm{FALSE}}}

\newcommand{\CP}{{\mathrm{C}}}

\renewcommand{\Pr}[1]{\mathbb{P}\left(#1\right)}
\renewcommand{\E}[1]{\mathbb{E}\left(#1\right)}

\newcommand{\Pcl}{\widetilde{H}_{n,p,k}}

\newtheorem{theorem}{Theorem}[section]

\newtheorem{corollary}[theorem]{Corollary}
\newtheorem{proposition}[theorem]{Proposition}
\newtheorem{lemma}[theorem]{Lemma}
\newtheorem{claim}[theorem]{Claim}

\numberwithin{equation}{section}
\numberwithin{firsttheorem}{section}

\title{On the Insertion Time of Cuckoo Hashing \footnote{This paper was written while the first two authors were affiliated with Max-Planck-Institute for Informatics, and while the third author was on sabbatical leave at the Max-Planck-Institute for Informatics. 
It has been accepted for publication in the \emph{SIAM Journal on Computing}. }}

\author{Nikolaos Fountoulakis \\
\small University of Birmingham, School of Mathematics  \\
\small Edgbaston, B15 2TT, United Kingdom \\
\small \texttt{n.fountoulakis@bham.ac.uk} \\
\and Konstantinos Panagiotou \\ 
\small Ludwig-Maximilians-Universit\"at, Mathematics Institute \\
\small Theresienstr.\ 39, 80333 M\"unchen, Germany \\
\small \texttt{kpanagio@math.lmu.de} \\
\and Angelika Steger \\
\small Institute of Theoretical Computer Science \\
\small ETH, 8092 Zurich, Switzerland \\
\small \texttt{steger@inf.ethz.ch}
}

\date{\today}

\begin{document}

\maketitle
\begin{abstract}
Cuckoo hashing is an efficient technique for creating large hash tables with high
space utilization and guaranteed constant access times. There, each item can be
placed in a location given by any one out of $k$ different hash functions. In this
paper we investigate the \emph{random walk} heuristic for inserting in an
online fashion new items into the hash table. Provided that~$k \ge 3$ and that the
number of items in the table is below (but arbitrarily close) to the theoretically
achievable load threshold, we show a polylogarithmic bound for the maximum insertion
time that holds with probability $1-o(1)$ as the size of the table grows large. 
\end{abstract}

\section{Introduction}

Hash tables are widely used in applications that need efficient data structures supporting the insertion, deletion, and lookup of elements. A key issue in the use of hash-tables is the handling of collisions. One particular technique that has attracted quite a bit of attention is the so-called \emph{cuckoo hashing}, which is based upon the paradigm of the power of many choices
\cite{abku99,mrs00}. 

The term cuckoo hashing was coined by Pagh and Rodler in~\cite{inp:pr01}.  In this work we consider a generalization of it, as it was described by Fotakis, Pagh, Sanders and Spirakis in~\cite{inc:fpss03}.
We are given a table $T$ with~$n$ locations, and we assume that each location can hold one item. Further settings where two or more items can be stored in each location have also been studied, see e.g.\ \cite{ar:csw07,ar:dw07,ar:fr07,inp:fkp11}, but we will not treat
those cases here. Moreover, we are given $k\ge 2$ hash functions~$h_1,\ldots, h_k$ that map any element from a universe~$U$ of items to $k$ (not necessarily distinct) positions in the table~$T$. Each given item can be stored at any position dictated by one out of those $k$ functions.
As in many previous works, we also assume that $h_1,\ldots, h_k$ are \emph{independent} and \emph{uniformly random} functions~$h_i:U\rightarrow T$; this assumption, though not being satisfactory from a practical point of view, is necessary in several of our arguments.

In his survey~\cite{inc:m09} Mitzenmacher outlined a number of open problems that remain to be solved in order to improve our understanding of the power of cuckoo hashing from a theoretical point of view. Among these is the issue of space utilization and the search for good upper bounds for the time needed to insert new elements. The first point was recently solved independently by the first two authors~\cite{fp10,fp10a} as well as Frieze and Melsted~\cite{inc:fp09}. The aim of this paper is to address the second point. Before we present our results we first need to outline the results from \cite{fp10,fp10a,inc:fp09} in more detail. 

\paragraph{Load Thresholds} A natural question in cuckoo hashing is the following. Let us denote by $I \subset U$ the set of available items. As $|I|$ increases, it becomes more and more unlikely that all items can be inserted into the table so that each item $i$ is assigned to one of the~$k$ locations $h_1(i), \dots, h_k(i)$. In other words, if~$|I|$ is ``small'' compared to $n=|T|$, then with high probability\footnote{i.e., with probability $1 - o(1)$ as $n\to \infty$}, there is an assignment to the locations of the table that respects the~$k$ choices of each item. On the other hand, if~$|I|$ becomes ``large'', then such an assignment does not exist with high probability (trivially, this happens at the latest when~$n+1$ items are available). The important question is whether there is a critical size for~$I$ where the probability for the existence of a valid assignment drops abruptly in the limit as $n\to \infty$ from 1 to 0, i.e., whether there is a \emph{load threshold} for cuckoo hashing. More precisely, we say that 
a value $c_k^*$ is the load threshold for cuckoo hashing with $k$ choices for each element if
\begin{equation}
\label{eq:threshold}
	\mathbb{P}
	\left(
	\begin{array}{c}
	\text{there is an assignment of $\lfloor cn\rfloor$ items to a table with} \\
	\text{$n$ locations that respects the choices of all items} \\
	\end{array}
	\right)
	\stackrel{n \to\infty}{\to}
	\begin{cases}
	1, \quad \text{if $c < c_k^*$}, \\
	0, \quad \text{if $c > c_k^*$}
	\end{cases}.
\end{equation}
In the case~$k=2$ there is a natural connection with random graphs. Indeed, we may think of the~$n$ locations of $T$ as the vertices of the graph, and of the items as edges that encode the two choices. If $|I|=m$ we obtain the classical Erd{\H{o}}s-R{\'e}nyi random (multi-)graph~$G_{n,m}^*$. 
The properties of this random graph are essentially those of the random graph $G_{n,m}$ on $n$ vertices and $m$ distinct uniform random edges.
Moreover, it easy to see by applying Hall's Theorem that~$G_{n,m}$ has no subgraph with more edges than vertices if and only if all items can be assigned to one of their preferred locations such that no location is assigned more than one item. It is well-known that the property ``$G_{n,m}$ has a subgraph with more edges than vertices'' coincides with the emergence of a \emph{giant connected component} that contains a linear fraction of the vertices, see e.g.~\cite{b:jlr00}. As the latter is known to happen when~$m$ crosses  $n/2$, we readily obtain that the load threshold for cuckoo hashing  and~$k = 2$ is at~$c_2^* = 1/2$. In other words, at most \emph{half} of the table can be filled in a way that respects the choices of all items.

The cases $k\ge 3$ are very different from $k = 2$. Tight results were obtained independently by the first two authors~\cite{fp10,fp10a}, as well as by Frieze and Melsted~\cite{inc:fp09}. Moreover, Dietzfelbinger et al.~\cite{unp:Dietz09} related the load threshold for cuckoo hashing to the satisfiability threshold of the $k$-XORSAT problem.  
\begin{theorem}
\label{thm:main}
For any integer $k \ge 3$ let $\xi^*$ be the unique solution of the equation
\begin{equation}
\label{eq:kxi}
	k = \frac{\xi^*(1-e^{-\xi^*})}{1-  e^{-\xi^*} - \xi^*e^{-\xi^*}}.
\end{equation}
Then $c_k^* = \frac{\xi^*}{k(1 - e^{-\xi^*})^{k-1}}$ is the load threshold for cuckoo hashing with $k$ choices per element. In particular, if there are $\lfloor cn \rfloor$ items, then the following hold with high probability.
\begin{enumerate}
\item If $c<c_k^*$, then there is an assignment of the items to a table with $n$ locations that respects the choices of all items.
\item If $c>c_k^*$, then such an assignment does not exist.
\end{enumerate}
\end{theorem}
Numerically we obtain for example that~$c_3^* \stackrel{.}{=} 0.917, c_4^* \stackrel{.}{=} 0.976$ and~$c_5^* \stackrel{.}{=} 0.992$, where ``$\stackrel{.}{=}$'' indicates that the values are truncated to the last digit shown. Moreover, a simple calculation reveals that~$c_k^* = 1 - e^{-k} + o(e^{-k})$ for~$k\to\infty$.

\paragraph{Fast Insertion}
Note that Theorem~\ref{thm:main} is non-algorithmic: it states that whenever the load of the hash-table is below the load threshold, then there {\em exists} with high probability an assignment such that each item is allocated to one of the $k$ positions given by the hash functions. The theorem does not, however, address the question whether such an assignment can actually be {\em found} efficiently. This is the problem that we address in this paper. 

More specifically, we study the distribution of the time that is needed to add an item into the table, assuming that some number of given items have already been allocated. We consider the following natural randomized insertion procedure, which is also known as the ``random-walk insertion heuristic''.
When we want to add a new item $i$ to the table, we first choose randomly one location among $h_1(i), \dots, h_k(i)$.  If that location is free, $i$ can be placed immediately and the algorithm terminates. If not, the item occupying the chosen location has to be displaced and is moved to a random location among its remaining $k-1$ choices to make room for $i$. This item in turn may need to displace another item, and so on. Consequently, inserting an item will require a sequence of moves, each maintaining the property that every item remains in one of its $k$ potential locations, until no further displacements are needed. The algorithm is described formally in the next section.

The random-walk insertion heuristic was introduced and simulated in~\cite{inc:fpss03} as a randomized counterpart of the deterministic breadth-first search for a free location to insert the new item. We refer the reader to that paper for a detailed discussion about the origins of the method and many references to related work. Note that the random-walk insertion heuristic may fail in several different ways. In particular, if it is not possible to add the new item to the table, it will get stuck in an infinite sequence of displacements of items. Moreover, since the heuristic acts only 'locally', there is no way of detecting such an undesirable situation. On the other hand, Frieze, Melsted and Mitzenmacher~\cite{inp:fmm09} studied this algorithm and showed that such 'bad' situations occur only very rarely. In particular, they showed that the running time, i.e., the number of displaced items, is polylogarithmic with high probability provided $k\ge 8$ and the load of the hash-table is not too close to the threshold $c_k^*$.

The main theorem of our paper states that the random-walk insertion heuristic actually succeeds in polylogarithmic time with high probability for any number of inserted items arbitrarily close to the load threshold and for all $k\geq 3$. 
\begin{theorem} \label{thm:insertion}
For any $\zeta > 0$ the following is true. Let $0 < \eps < 1$, and for $k\geq 3$ set $c = {k+\log(k-1)\over (k-1)\log(k-1)}$. For any set of $m=\lfloor (1-\eps)c_k^* n\rfloor$ items of the universe $U$, if the hash functions $h_1,\ldots, h_k$ are independent and uniformly distributed over a table with $n$ locations, then with probability $1-o(1)$ each of the items will be inserted 
into the table in time $O(\log^{2+c+\zeta} n)$. 
\end{theorem}
Observe that $c \stackrel{.}{=} 2.66$ for $k = 3$, $c \stackrel{.}{=} 1.54$ for $k=4$ and $c \stackrel{.}{=} 1.15$ for $k = 5$. Moreover, a simple calculation shows that $c  = \frac1{\log k} + O(\frac1k)$ as $k$ grows. Our exponent in the bound of the running time thus also improves upon that from \cite{inp:fmm09}, which is greater than $2 + 2c$.

As a last remark, note that for $k = 2$ the random-walk insertion heuristic is a deterministic algorithm. The total running time to insert a given number of elements was studied in the papers by Devroye and Morin~\cite{ar:dm03} and in great detail by Drmota and Kutzelnigg~\cite{ar:dk12}. In particular, these papers show that the running time is linear in the number of inserted elements. We conjecture that for $k \ge 3$ this should also be the case.

\paragraph{Outline} Our proof of Theorem~\ref{thm:insertion} develops further some ideas from~\cite{inp:fmm09} and combines them with several new structural properties of random hypergraphs. In Section~\ref{sec:mainProof} we introduce some basic facts and relate cuckoo hashing to random hypergraphs and orientations of their edge sets. A general outline of the argument can be found there. In particular, it turns out that a crucial property needed in the analysis is that such a random hypergraph has no ``dense spots'', i.e., no subgraphs that have a density (fraction of edges to vertices) that is much larger than the density of the whole graph, c.f.\ Theorem~\ref{thm:density}. In Section~\ref{proof:density}, which is formulated entirely in random graph jargon, we prove Theorem~\ref{thm:density}. Finally, in Section~\ref{proof:expansion} we prove some additional auxiliary properties of random hypergraphs.

\section{The Insertion Algorithm and its Analysis}
\label{sec:mainProof}

\subsection{Random-Walk Insertion \& Random Hypergraphs}
\label{ssec:rwi}
%

We begin with a formal description of the random-walk insertion heuristic. We assume that there are~$k\geq 3$ hash functions~$h_1,\ldots, h_k$ that map a universe~$U$ to the $n$ locations of a hash table~$T$. We denote by~$T(i)$ the contents of the~$i$th location of~$T$, and we write~$T(i)=\emptyset$ if the~$i$th location is empty. Also, if~$e$ is an item that has been inserted into the table, we denote by $I(e)$ the index of the hash function that~$e$ currently uses, i.e., $T(h_{I(e)}(e))=e$. With these definitions at hand we are able to present in the following table a formal description of the insertion algorithm exposed in the introduction.
\begin{table}[h]
\label{t:algo}
\centering 
\begin{tabular}{c l}
\hline 
   & {\bf algorithm} {\sc Insert}($T$, $I$, $e$) -- Insert $e$ into $T$ using the random-walk heuristic\\
1 & $\Suc \leftarrow  \Fa$; \\
2 & $j\leftarrow 0$; \\
3 & {\bf repeat} \\
4 & \hspace{1cm} Choose uniformly at random~$i \in \{1,\ldots, k\}\setminus \{j\}$; \\
5 & \hspace{1cm}~$I(e) \leftarrow i$; \\
6 & \hspace{1cm} {\bf if}~$T(h_i(e)) \not = \emptyset$ {\bf then} \\
7 & \hspace{1.5cm}~$e'\leftarrow T(h_i(e))$; \\
8 & \hspace{1.5cm}~$j \leftarrow I(e')$; \\
9 & \hspace{1.5cm}~$T(h_i(e)) \leftarrow e$; \\
10 & \hspace{1.5cm}~$e\leftarrow e'$; \\
11 & \hspace{1cm} {\bf else} \\
12 &\hspace{1.5cm}~$T(h_i(e)) \leftarrow e$; \\
13 &\hspace{1.5cm}~$\Suc \leftarrow \Tr$; \\
14 &\hspace{1cm} {\bf endif} \\
15 & {\bf until~$\Suc = \Tr$} \\
\hline
\end{tabular} 
\end{table}

Our analysis of this algorithm begins with describing the allocation of elements in terms of certain operations on hypergraphs. More specifically, the hash table of size~$n$ corresponds to a set of~$n$ vertices, and the $k$ locations $h_1(i), \dots, h_k(i)$ of an item $i$ correspond to a hyperedge of size (at most)~$k$. As we assume that the hash functions are truly random, a set of~$m$ elements gives rise to a~$k$-uniform random multi-hypergraph $H_{n,m,k}^*$ with vertex set $V_n = \{1, \dots, n\}$ and~$m$ edges, each one chosen independently and uniformly at random with replacement among all $k$-multisubsets of the vertex set. Note that in this definition of~$H_{n,m,k}^*$ we actually interpret the word ``multi" in two ways: {\em (i)} an edge has size~$k$, but may contain a particular vertex several times, and {\em (ii)} the edge set of~$H_{n,m,k}^*$ may be a multiset, i.e.,\ a particular edge may occur multiple times.

With slight abuse of terminology, we say that a multi-hypergraph~$H=H(V,E)$ with vertex set $V$ and edge set $E$ is a~$k$-\emph{graph} if it is~$k$-bounded, that is, every hyperedge is a subset of~$V$ with at most~$k$ vertices. Observe, that the random hypergraph~$H_{n,m,k}^*$ corresponds in a natural 
way to a~$k$-graph, by projecting each ordered~$k$-tuple of vertices that forms an edge in~$H_{n,m,k}^*$ to the set of vertices contained 
in this $k$-tuple. In what follows, we will be using the symbol~$H_{n,m,k}^*$ to denote both objects; each time the interpretation should 
be clear from the context. 

\paragraph{Notation} Let $H = H(V,E)$ be a multi-hypergraph. In the rest of the paper we will use the following notation. The \emph{density} of $H$ is the ratio $|E|/|V|$. Moreover, for any $V'\subseteq V$ we write~$H[V']$ for the subgraph of~$H$ \emph{induced by} $V'$, i.e., $H[V']$ has vertex set $V'$ and its edges are those edges in $E$ that contain vertices only from $V'$. We will also write $E_H(V')$ for the edge set of $H[V']$ and $e_H(V') = |E_H(V')|$. We will omit the subscript if the graph we refer to is clear from the context. Finally, slightly abusing notation, we will say that the density of $V'$ is $e_H(V')/|V'|$, again if the reference to $H$ is clear from the context.


\subsection{Orientations and the~$o$-neighborhood of a Vertex} 

For a~$k$-graph~$H=H(V,E)$ with~$|E| \leq |V|$ 
an injective mapping~$o:E\rightarrow V$ such that~$o(e) \in e$ for all~$e\in E$ is called an \emph{orientation}. For~$e\in E$ and~$v\in e$ we will say that~$e$ is oriented to~$v$ if~$o(e) = v$. Observe that in the setting of cuckoo hashing an orientation corresponds to a proper assignment of the items to locations in the hash table, i.e., such that each location contains at most one item, and each item is assigned to a location that is prescribed by one of the $k$ hash functions.
Similarly, the random-walk insertion heuristic can be viewed as a random walk on the vertex set of the corresponding~$k$-uniform hypergraph. If we want to stress that the random walk starts with a particular assignment that corresponds to an orientation~$o$, we also speak of an \emph{$o$-random walk} on the vertex set of the hypergraph~$H$.

Before we proceed with the proof of Theorem~\ref{thm:insertion} in the next section we need to collect some basic properties of orientations and fix some necessary notation. Given an orientation~$o$ of a hypergraph~$H$, we denote by~$F_o(H)$ the set of \emph{free} vertices, that is, vertices to which no edge is oriented to. We call the remaining vertices \emph{occupied}. Furthermore, we will define certain quantities that describe how 'far' a given vertex is from~$F_o(H)$. Formally, for any~$v \in V$ let the \emph{first}~$o$-\emph{neighborhood}~$\mathcal{\overline N}_{o,1}(v; H)$ be the set of vertices in the edge oriented to~$v$, excluding~$v$ itself, i.e,
\[
	\mathcal{\overline N}_{o,1}(v; H) = \{ u \in e~:~ o(e) = v \text{ and } u \neq v\}.
\]
More generally, the~$t$th~$o$-neighborhood of~$v$ contains all vertices in any edge that is oriented to some vertex of the~$(t-1)$st~$o$-neighborhood of~$v$, excluding all vertices in previous neighborhoods. Formally, we have
\[
	\mathcal{\overline N}_{o,t}(v; H) = \big\{ u \in e~:~ o(e) \in \mathcal{\overline N}_{o,t-1}(v;H) \big\} \setminus \bigcup_{j=1}^{t-1}\mathcal{\overline N}_{o,j}(v;H).
\]
We also define the~$0$th~$o$-neighborhood of~$v$ to be~$v$ itself. Moreover, we write $\N_{o,t}(v;H) = \cup_{j=0}^t \mathcal{\overline N}_{o,j}(v;H)$ and $N_{o,t}(v;H) = |\N_{o,t}(v;H)|$. We will omit the reference to $H$ if the hypergraph that we refer to is clear from the context. Finally, if~$v$ is a vertex of~$H$ and~$S\subseteq V$, then the~$o$-distance of~$S$ from~$v$ is defined by
\[
	d_o(v,S; H) = \min \{t  :  \N_{o,t}(v) \cap S \not = \emptyset \},
\]
i.e., it is the smallest $t$ such that $\N_{o,t}(v)$ contains some vertex is $S$. Note that the $o$-distance provides a way to measure the complexity, i.e., the number of items that have to be displaced, of an $o$-random walk. Indeed, if $d_o(v, F_o) = d$, then inserting a new item at location $v$ requires (deterministically) the displacement of at least $d$ items.

Note also that for~$t\geq 1$ the definition implies that the~$t$th neighborhood of any vertex can contain at most~$(k-1)^t$ vertices. Thus, by summing up a geometric series we get the bound
\begin{equation} \label{eq:tlevel}
 N_{o,t}(v) \leq (k-1)^{t+1}. 
\end{equation}  

\subsection{Proof of Theorem~\ref{thm:insertion}: Analysis of the Insertion Algorithm}

Before we proceed with the actual proof details, let us give a rough sketch of the main argument. Assume that~$0 \le m \le \lfloor (1-\eps)c_k^*n \rfloor$ items~$I_1, \dots, I_m \in U$ are given, together with the values of the~$k$ hash functions. Then the above discussion implies that the corresponding hypergraph is distributed like~$H_{n,m,k}^*$, where the~$m$ items correspond to the~$m$ edges. Let us fix the realization of~$H_{n,m,k}^*$ that we denote by~$H$ for brevity. We will bound the insertion time under the assumption that the hypergraph satisfies some properties holding with high probability that we will state explicitly. In particular, these properties will guarantee that~$H$ and (thus) all of its subgraphs have an orientation. Let~$H_j = H_j(V_n, E_j)$ be the subgraph of~$H$ on the same vertex set, which includes only the edges that correspond to the first~$j$ items~$I_1, \dots, I_j$, and let~$o_j:E_j\to V_n$ be \emph{any} orientation of~$H_j$. Subsequently, we will argue that if~$H$ has this set of desired properties, then with probability at least~$1 - O(n^{-1 - \zeta/2})$ the $o_j$-random walk will displace at most~$O(\log^{2+c+\zeta}n)$ items before inserting~$I_{j+1}$ successfully. This implies the statement of the theorem.

Suppose that the edge corresponding to~$I_{j+1}$ is initially allocated to~$v \in V_n$, and assume that~$v$ is occupied; otherwise there is nothing to show. Following the arguments in~\cite{inp:fmm09} we will consider a decomposition of the vertex set~$V_n$ into two sets according to the~$o_{j}$-distance of~$F_{o_j}$ from each vertex. In particular, for some~$C>0$, we let~$S \subseteq V_n$ be the set of vertices~$v \in V_n$ such that~$d_{o_j}(v,F_{o_j}) \leq C$ and let~$B=V_n\setminus S$. 
In the work of Frieze et al.~\cite{inp:fmm09} the parameter~$C$ was of order~$\log \log n$. In contrast, for our analysis, it is sufficient that $d_{o_j}(v,F_{o_j})$ is bounded, and this has also the additional benefit that it allows us to shave off some logarithmic factors from the resulting running time. We show that $C$ can be chosen such that~$S$ covers almost all of the hypergraph. 

Note that if~$v \in S$, then the definition of $S$ implies that there is at least one free vertex within~$o_j$-distance~$C$ from~$v$, and the~$o_j$-random walk will thus hit a free vertex with probability at least~$1/(k-1)^C$ within the next $C$ steps. In order to treat the case $v \not \in S$ we will first show that certain expansion properties of~$H$ (that hold with high probability for $H_{n,m,k}^*$) guarantee that the~$o_j$-neighborhood up to $o_j$-distance roughly $\log_{k-1} n$ from~$v$ grows almost like a~$(k-1)$-regular hypertree. This, in turn, will then allow us to show that for vertices $v \not \in S$ the~$o_j$-random walk will hit~$S$ with reasonably high probability after a logarithmic number of steps.

Our plan for the proof of Theorem~\ref{thm:insertion} is thus as follows. In the next two subsections we define two properties of~$k$-graphs, a density and an expansion property, and show that a random hypergraph $H^*_{n,m,k}$ has these properties with high probability. We also show that the density property implies that the set $S$ is large, and that the expansion property implies that a random walk starting in a vertex not in $S$ will hit $S$ within a logarithmic number of steps with high probability. In Section~\ref{ss_proof} we then show how these two properties conclude the proof of Theorem~\ref{thm:insertion}.

\subsubsection{Density Properties}

Assume that~$H=H(V,E)$ is a~$k$-graph.
We define the following property.
 
\medskip 

\noindent
{\bf Property $D_{\delta}$}: For all $\emptyset \neq V' \subseteq V$ we have 
$$ e(V') < (1-\delta) |V'|.$$ 
 
\noindent Note that if a $k$-graph has property $D_\delta$ for some $\delta > 0$, then there exists an orientation as well; this follows immediately from Hall's theorem. The next proposition states that for \emph{any} such orientation $o$, most of the vertices are such that the set of free vertices $F_o$ is within bounded $o$-distance from them. This will be an important ingredient in the proof of Theorem~\ref{thm:insertion}, since for such vertices, the probability that the random-walk insertion heuristic finds the shortest sequence of displacements to some free vertex is bounded from below by a constant.
\begin{proposition} \label{prop:decomp} 
Let $H=H(V,E)$ be a $k$-graph that has Property $D_{\delta}$ for some $\delta > 0$. Let $o:E\to V$ be any orientation. Then, for all $\alpha >0$, there exists $C=C(\alpha, \delta)>0$ and a set $S \subseteq V$ with $|S| \ge (1-\alpha) |V|$ such that for every $v \in S$ we have $d_o(v,F_o)\leq C$.
\end{proposition}
\noindent
We defer the proof to Section~\ref{proof:decomp}. The next theorem states that~$H_{n,m,k}^*$ has Property~$D_{\delta}$ with high probability for some suitable $\delta>0$.
\begin{theorem}  \label{thm:density}
Let $0<\eps <1$, $k\ge 3$ and suppose that $m=\lfloor (1-\eps) c_k^*n \rfloor$. There exists a $\delta = \delta(\eps,k)>0$ such that 
$H_{n,m,k}^*$ has property~$D_{\delta}$ with probability~$1-o(1)$.  
\end{theorem}
\noindent
We prove this theorem in Section~\ref{proof:density}. 
Together with Proposition~\ref{prop:decomp} this gives us a statement about the typical structure of~$H_{n,m,k}^*$. 
\begin{corollary} \label{lem:decomp}
Let~$\eps  > 0, \alpha >0$ and assume that~$m=\lfloor (1-\eps) c_k^* n\rfloor$. Then there is a~$C=C(\alpha, \eps) >0$ such that \emph{every} subgraph~$\hat{H} = \hat{H}(V_n, \hat{E})$ of~$H_{n,m,k}^*$ has the following property with high probability. For \emph{any} orientation~$\hat{o}$ of~$\hat{H}$ there exists a set~$S \subseteq V_n$ such that~$|S|\ge(1-\alpha)n$ and every~$v \in S$ satisfies~$d_{\hat{o}} (v,F_{\hat{o}}; \hat{H}) \leq C$.
\end{corollary}

\subsubsection{Expansion Properties}

The density properties described in the previous section will allow us to deal in the proof of Theorem~\ref{thm:insertion} with insertions of new items at vertices that are at bounded $o$-distance from the set of free vertices. In order to deal with insertions at any other vertex we will consider a different set of structural properties of hypergraphs.

For a set of edges~$E'$ of a hypergraph~$H=H(V,E)$, we denote by~$V(E')$ the set of vertices contained in edges of~$E'$. 
We say that a~$k$-graph~$H=H(V,E)$ has \emph{expansion property}~$\mathcal{E}$, if it satisfies the following two conditions.

\medskip 

\noindent
{\bf Property~$\mathcal{E}$}:
\begin{enumerate}
\item[1.]  For all~$E'\subseteq E$ with~$\log \log |V| < |E'| < |V|/k$ we have 
$$ |V(E')| \geq (k-1-x_{|E'|}) |E'|,\qquad\mbox{
where~$x_s = {\log_k((k-1)e^k) \over \log_k(|V|/s) - 1}$}.$$
\item[2.] For all~$E'\subseteq E$ with 
$|E'| \le \log \log |V|$ we have 
$$ |V(E')| \geq (k-1) |E'|.~$$
\end{enumerate}
Note that for any~$E'\subseteq E$ we have~$|V(E')| \le k |E'|$, and moreover, if the graph with edge set~$E'$ is connected, then even~$|V(E')| \le (k-1)|E'|$. Property~$\mathcal{E}$ thus guarantees that~$|V(E')|$ is 'close' to this upper bound for all~$E'\subseteq E$, i.e., the graph around any given vertex expands rapidly.

Certainly, the choice of the parameters in the definition of~$\mathcal{E}$ is somehow arbitrary and not best possible. Nevertheless, they suffice for our arguments. In Section~\ref{proof:expansion} we show the following statement.
\begin{proposition} \label{prop:expand}
If~$m\leq c_k^* n$, then~$H_{n,m,k}^*$ has Property~$\mathcal{E}$ with high probability. 
\end{proposition}
Note also that if $H$ has Property $\mathcal{E}$, then all its subgraphs have also this property. Given an orientation~$o$ of a hypergraph and a vertex~$v$, recall that~$N_{o,t}(v)$ is the number of vertices 
within~$o$-distance at most~$t$ from~$v$. In Section~\ref{proof:expansion} we show the following statement, which quantifies how fast $N_{o,t}$ grows as a function of $t$, if the given graph has Property $\mathcal{E}$.
\begin{lemma} \label{lem:growth}
Let~$\zeta>0$, $k\geq 3$ and $c = c(k)$ be as in Theorem~\ref{thm:insertion}. Then, for any~$\beta >0$ sufficiently small and for any $n$ sufficiently large the following is true. Let~$H=H(V,E)$ be a~$k$-graph with~$n$ vertices that has Property~$\mathcal{E}$, and let~$o:E\to V$ be any orientation. Set~$T=\log_{k-1} n + (c+\zeta) \log_{k-1} \log_{k-1} n$. If~$v \in V$ is such that 
there is no free vertex within~$o$-distance~$T$ from~$v$, then~$N_{o,T} (v) > \beta n$.
\end{lemma}
The proof can be found in Section~\ref{proof:expansion}. In other words, this lemma states that number of vertices with $o$-distance $T$ from a given vertex $v$ is $\Omega(n)$, provided that there is no free vertex within this $o$-distance. Intuitively, placing a new item at such a vertex $v$ dominates the running time of the random-walk insertion heuristic, since in this case it is forced to displace (in the best case) at least $T$ items.

In order to show a similar statement about $N_{o,t}$ for vertices that are 'close' to the set of free vertices we perform a simple technical modification. Given a hypergraph~$H=H(V,E)$ and an orientation~$o$ of its edges, we define an auxiliary hypergraph~${H'} = {H'}(V',E')$ by replacing every free vertex of~$H$ by a~$(k-1)$-regular hypertree of depth~$T$ (on a new set 
of vertices) rooted at this vertex. We extend~$o$ to an orientation~$o'$ of the edges of~${H'}$ by orienting each new edge towards the root of its tree. Thus, the union of the leaves of each such tree are the free vertices of~$H'$. 
\begin{corollary} \label{lem:expansion}
With the same assumptions as in Lemma~\ref{lem:growth} and $H'$ as in the previous paragraph we obtain for \emph{all}~$v\in V$ that~$N_{o',T}(v; H') > \beta n$.
\end{corollary} 

\subsubsection{Proof of Theorem~\ref{thm:insertion}}\label{ss_proof}

In this section we combine the statements of Corollary~\ref{lem:decomp} and Corollary~\ref{lem:expansion} to derive a high probability bound for the running time of the random-walk insertion heuristic. Assume that~$m=\lfloor (1-\eps) c_k^* n\rfloor$. From Theorem~\ref{thm:density} we know that there exists a~$\delta = \delta(\eps,k)>0$ such that~$H_{n,m,k}^*$ satisfies Property~$D_{\delta}$ with high probability. In particular,~$H_{n,m,k}^*$ has an orientation with high probability, and so do all of its subgraphs. From Proposition~\ref{prop:expand} we also know that~$H_{n,m,k}^*$ satisfies Property~$\mathcal{E}$ with high probability; similarly, all its subgraphs have Property~$\mathcal{E}$. We thus may assume that we begin with a certain realization of~$H_{n,m,k}^*$ that has both properties, which, for
convenience, we abbreviate by~$H$. Let~$\hat{H} = \hat{H}(V_n, \hat{E})$ be \emph{any} subgraph of~$H$ on the same vertex set. Fix \emph{any} orientation~$\hat{o}$ of its edges, and any~$\zeta > 0$. Thereafter, choose~$\alpha$ small enough so that Corollary~\ref{lem:expansion} can be applied (to~$\hat H$) with $\beta = 2\alpha$ and set~$C=C(\alpha ,\eps) > 0$ as in Corollary~\ref{lem:decomp}. 
Note that this also specifies the set~$S$, i.e.,~$S \subseteq V_n$ is such that~$|S|\ge(1-\alpha)n$ and every~$v \in S$ satisfies~$d_{\hat{o}} (v,F_{\hat{o}}; \hat{H}) \leq C$.
\begin{lemma} \label{lem:SingleStage}
Assume that~$\alpha, \zeta, C$ and~$H,\hat{H},\hat{o}$ are specified as above. Let~$T$ and~$c$ be defined as in Corollary~\ref{lem:expansion} and Theorem~\ref{thm:insertion}. Then
the probability that a~$\hat{o}$-random walk on~$\hat{H}$ starting at a vertex~$v_0$ hits a free vertex within~$T+C$ steps is at least 
${\alpha\,(k-1)^{-C}}\,{(\log_{k-1} n)^{-c-\zeta}}$.
\end{lemma}
\begin{proof}
Let us consider the following stopping rule. 
Starting from~$v_0$, we walk either for~$T$ steps or until we hit a free vertex, whatever occurs earlier. 
If the latter is not the case, then we walk for~$C$ additional steps or until we hit a free vertex.
We consider the same~$\hat{o}$-random walk on~$\hat{H}'$, that is, 
the~$\hat{o}$-random walk that starts at~$v_0$ and makes the same random choices as the one in~$\hat{H}$, and possibly new ones, if the random walk on~$\hat{H}$ stopped at a free vertex. Let~$u$ be the vertex that has been reached after~$T$ steps. 
The growth property guaranteed by Corollary~\ref{lem:expansion} implies that there are at least~$\beta n - \alpha n \ge \alpha n$ vertices within~$\hat{o}'$-distance~$T$ from~$v_0$ in~$\hat{H}'$ that either belong to~$S$ or to one of the trees we added to~$\hat{H}$. The probability of hitting such a vertex after~$T$ steps is at least~${\alpha n  \, (k-1)^{-T}}$. If~$u \not \in S$, then~$u$ belongs to one of the trees that we added to~$\hat{H}$, and we conclude that the~$\hat{o}$-random walk on~$\hat{H}$ has stopped before making~$T$ steps. If~$u \in S$, then we stop the~$\hat{o}$-random walk in~$\hat{H}'$ but we continue it in~$\hat{H}$ for another~$C$ steps or until a free vertex is found. The probability that we end up at a free vertex is at least~$(k-1)^{-C}$, since the assumption that~$u$ belongs to~$S$ implies that there is at least one free vertex within $\hat{o}$-distance $C$ from $u$. Thus the probability of success is at least ${\alpha n  \, (k-1)^{-T-C}}$, and the proof is completed.
\end{proof}
To conclude the proof of Theorem~\ref{thm:insertion}, we split the $o$-random walk into \emph{phases}, where a phase is either a window of duration~$T+C$ or until a free vertex was hit. We repeatedly use the above lemma to bound the number of unsuccessful phases, since it applies to any subgraph of~$H^*_{n,m,k}$ and any orientation of it. More precisely, suppose that we are given any orientation~$\hat{o}$ of some subgraph~$\hat{H}$ of~$H^*_{n,m,k}$, and we want to extend~$\hat{o}$ to the graph~$\hat{H}$ plus some edge~$e$ that is in~$H^*_{n,m,k}$ but not in~$\hat{H}$. (In the cuckoo hashing setting, this amounts to the situation that we have already inserted a specific number of items in the table, and we want to add another item.) Suppose that we set~$\hat{o}(e) = v_0 \in e$;~$v_0$ is chosen randomly among the vertices in~$e$, but this is not important. If the first phase is unsuccessful, then the above statement applies with the starting vertex being the vertex in which the previous phase ended, say~$v_1$, and with a new orientation~$\hat{o}_1$ of the hypergraph~$\hat{H}$ with~$e$ and without the edge that was oriented to~$v_1$ before the last step of the~$\hat{o}$-random walk.

The above arguments imply that for any $\zeta >0$ the probability that at least 
$\log_{k-1}^{c+1 + 2\zeta}n$ phases are unsuccessful given that a new item is inserted 
in $v_0$ is $O(1/n^{1+ \zeta})$. As each phase lasts for $T+C = O(\log n)$ steps, we deduce
that with probability $1-O(1/n^{1+ \zeta})$, the random walk inserts the new item within $O(\log^{2+c + 2\zeta} n)$ steps. 
As the total number of inserted elements is $O(n)$ this concludes the proof of Theorem~\ref{thm:insertion}.

\paragraph{Remark} Note that the same argument yields the bound $O(\log^{1+c + 2\zeta} n)$ for the \emph{expected} total number of steps performed by one $o$-random walk, since the number of unsuccessful phases is dominated by a geometric distribution.

\subsection{Proof of Proposition~\ref{prop:decomp}: $o$-neighborhoods and the Density of a Hypergraph} \label{proof:decomp}

The main idea of the proof is as follows. As $H$ has Property~$D_{\delta}$ we know that $|F_o| = |V| - |E| >  \delta |V|$. Suppose that we remove all free vertices and all edges that contain a free vertex. As the removal of any edge~$e$ that contains a free vertex generates a new free vertex, namely $o(e)$, this process generates a subhypergraph of $H$, which is induced by $V\setminus F_o$, with a new set of free vertices. For this subhypergraph we can again use Property~$D_{\delta}$ to deduce that the number of free vertices is at least a $\delta$-fraction of the number of vertices. We repeat this stripping until we are left with less than $\alpha |V|$ vertices -- and show that a constant number of rounds suffice.

More formally, let $F_0 := F_o$ and $L_0:= V$. Define inductively $L_{i+1} = L_i \setminus F_i$ and let $F_{i+1}$ be the set of free vertices in the hypergraph induced by $L_{i+1}$. Since $H$ has Property~$D_{\delta}$, it follows that $|F_0| > \delta |V| = \delta |L_0|$. We claim that for all $i\geq 0$ we have $|F_{i+1}| \geq \delta |L_{i+1}|$ as well. The two crucial observations are:
\begin{itemize}
\item[\em (i)] $|L_{i+1}|=e(L_i)$, as $L_{i+1}$ contains exactly the vertices that are the images (under~$o$) of the edges in the hypergraph induced by $L_i$.
\item[\em (ii)] $|F_{i+1}| = e(L_i) - e(L_{i+1})$, as every edge that belongs to the hypergraph induced by $L_i$ but not to the one induced by $L_{i+1}$ generates exactly one free vertex in $F_{i+1}$. 
\end{itemize}
As $H$ has Property~$D_{\delta}$ we also know that $e(L_{i+1}) < (1-\delta) |L_{i+1}|$. Hence,
\begin{equation*} 
\begin{split}
|F_{i+1}| \stackrel{(ii)}{=} e(L_i) - e(L_{i+1}) > e(L_i) - (1-\delta) |L_{i+1}| \stackrel{(i)}{=} \delta |L_{i+1}|.
\end{split}
\end{equation*} 
Thus 
\begin{equation*} 
\begin{split}
|L_{i+1}| = |L_i| -|F_i| < (1-\delta)|L_i|\qquad\text{for all $i\ge 0$.}
\end{split}
\end{equation*} 
To conclude the proof observe that this implies that for all $t\geq 1$
$$ |L_t| < (1-\delta)^t |L_0| = (1-\delta)^t |V|. $$
Thereby, choosing $t= \lceil \log_{1-\delta} \alpha \rceil$ we deduce that $|L_t| < \alpha |V|$. Thus, we may take $S:= V\setminus L_t$  
and  $C(\alpha, \delta) = \lceil \log_{1-\delta} \alpha \rceil$. 

\section{Proof of Theorem~\ref{thm:density}: the Subgraph Density of $H_{n,m,k}^*$} \label{proof:density}

Towards the proof of Theorem~\ref{thm:density} we begin with showing a lemma that will enable us to restrict our arguments to a specific subgraph of $H_{n,m,k}^*$. In particular, we will consider the so-called \emph{core}, which is the unique maximum subhypergraph of a given hypergraph $H$ with minimum degree at least two. Of course, the core of $H$ might contain no vertex.  We will denote the core of $H$ by $\CP(H)$.

The core is a well-studied object in the literature of random graphs and it has several applications in the analysis of algorithms, for example in the context of load balancing and hashing~\cite{ar:csw07,ar:fr07,fp10} or graph coloring~\cite{ar:ct04}. In particular, the distribution of the number of vertices and edges in $C(H_{n,m,k}^*)$ is well-known, see e.g.~\cite{ar:psw96,ar:c04,ar:m05,ar:k06}, and we will make extensive use of these results. 

A standard algorithm that reveals the core of a given hypergraph is the so-called \emph{stripping process} and works as follows. We repeatedly choose an arbitrary vertex of degree at most one and we remove it from the graph, together with the single edge that it contains (if any). This process stops when there are no more available vertices of degree at most one; what remains is either an empty hypergraph with no vertices (if the core is empty), or, otherwise, the core itself.
\begin{lemma} \label{lem:transfer}
Let $H=H(V,E)$ be a $k$-graph where every edge contains at least two distinct vertices such that 
\begin{enumerate}
\item[1.] $\CP(H)$ has Property~$D_{\delta}$ for some $0 < \delta < 1/4$ and
\item [2.] $H$ has Property $\mathcal{E}$.
\end{enumerate}
Then there exists $\zeta > 0$ such that $H$ itself has Property~$D_{\zeta \delta}$. 
\end{lemma}
\begin{proof}
Let $S \subseteq V$. We set $v_S = |S|$ and we write $e_S$ for the number of edges in~$H[S]$. Moreover, we write $d_S = e_S/v_S$ for the density of $H[S]$. We will show that for any $S \subseteq V$ we have $d_S < 1 -\zeta \delta$ for some $\zeta > 0$. Towards this goal we will make a case distinction depending on the number of edges $e_S$. 

The first part of Property~$\mathcal{E}$ implies that 
for all $\beta \in (0,1)$ there exists a $0< \gamma < 1/k$ such that for all $E'\subseteq E$ with 
 $\log \log |V| < |E'| \leq \gamma |V|$ we have 
$$ |V(E')| \geq (k-1-\beta) |E'|. $$
Setting $\beta =1/2$,  let $\gamma = \gamma (1/2)$ be as above.
Let us start with the case $e_S \leq \gamma |V|$. The bound on the density of $S$ can be deduced right away from 
Property~$\mathcal{E}$. Indeed, by using the first part of Property~$\mathcal{E}$, 
if $e_S > \log \log |V|$, then $v_S \geq (k-1-1/2)e_S$, yielding
\begin{equation}\label{eq:density1} 
d_S = {e_S \over v_S} \leq {1 \over k-1 -1/2} \stackrel{(k \ge 3)}{\leq} {2\over 3}.
\end{equation}
Using the second part of Property~$\mathcal{E}$, if $e_S\leq \log \log |V|$, then $v_S \geq (k-1)e_S$. A rearrangement yields 
\begin{equation} \label{eq:density13}
d_S = {e_S \over v_S} \leq {1\over k-1} \stackrel{ (k\geq 3 )}{\leq} {1\over 2}. 
\end{equation}
Suppose now that~$e_S > \gamma |V|$. We will make a further case distinction depending on the number of edges that belong to 
the core of~$H[S]$.
First, let us assume that there are at least~$\gamma |V|/2$ edges of~$H[S]$ that do not belong to the core of~$H[S]$. To avoid unnecessary complications we will assume that~$H[S]$ is connected; clearly it is sufficient to argue about connected sets. Consider the stripping process on~$H[S]$. This induces an ordering of the set of vertices of~$H[S]$ not belonging to the 
core; it is the ordering according to which these vertices are deleted from $H[S]$.
Let~$v_1,\ldots, v_t$ be this ordering for an integer~$t \le v_S$.

Note that whenever we delete a vertex of degree one during the stripping process, this is accompanied by an edge that is deleted too, that is, the (single) edge that this vertex is contained in.  Each one of the remaining~$k-1$ vertices of this edge either belongs to  the core of~$H[S]$ or it appears in some position after the deleted vertex in the specified ordering. Let us now consider the ordering in reverse and let~$i$ be the minimum index such that the number of edges that contain~$v_i,\ldots, v_t$ in the graph induced by $v_i, \dots, v_t$ and the vertex set of $\C(H[S])$ is $\lfloor \gamma |V|/2\rfloor$; there are~$s:=t-i+1$ vertices there. Assuming that among them there are~$x$ vertices that were isolated at the point in time they were removed by the stripping process, we have~$s=\lfloor\gamma |V|/2\rfloor + x$. The first part of Property~$\mathcal{E}$ implies that the number of vertices that are contained in these edges is at least $(k-1-\beta) \lfloor \gamma |V|/2\rfloor$. 
Among these vertices at least~$(k-1-\beta)\lfloor \gamma |V|/2\rfloor - s = (k-2 -\beta)\lfloor \gamma |V|/2\rfloor -x$ 
vertices must belong to the core of~$H[S]$. Let~$S_0$ 
denote these vertices and let~$S_1$ denote the remaining vertices of the core of~$H[S]$.  
In other words,~$|S_0| \geq (k-2 -\beta)\lfloor \gamma |V|/2\rfloor -x$.
As the core of~$H[S]$ is a subgraph of the core of~$H$, Property~$D_{\delta}$ implies that 
the core of~$H[S]$ contains at most~$(1-\delta)(|S_0| + |S_1|)$ edges. 
Now we can write an upper bound on the density of~$S$. We have 
\begin{equation*} 
\begin{split}
d_S = {e_S \over v_S} & \leq {t -x + (1-\delta) (|S_0| + |S_1|) \over t + |S_0| + |S_1|} = 
1 - {x \over t + |S_0| + |S_1|} - \delta {|S_0| + |S_1|\over t + |S_0| + |S_1|}.
\end{split}
\end{equation*}
Using that~$t+|S_0| + |S_1| \leq |V|$ we infer that
\begin{equation*}
\begin{split}
{e_S \over v_S} & \leq 1 - {x\over |V|}- \delta {|S_0| \over |V|} 
\leq 1 - {x\over |V|}- \delta {(k-2 -\beta)\lfloor \gamma |V|/2\rfloor -x \over |V|} \leq 1 -\delta {(k-2 -\beta)\lfloor \gamma |V|/2\rfloor  \over |V|}.
\end{split}
\end{equation*}
So,~$\lfloor \gamma |V|/2\rfloor\geq \gamma |V|/4$ implies that
\begin{equation}\label{eq:density2}
\begin{split}
{e_S \over v_S} &  
\leq
1 -\delta  {(k-2 -\beta) \gamma |V|\over 4|V|} = 1 -\delta  (k-2 -\beta) \gamma/4.  
\end{split}
\end{equation}
Finally, assume that less than $\gamma |V|/2$ edges of $H[S]$ do not belong to the core of $H[S]$. With~$e_{\CP (S)}$ and $v_{\CP(S)}$ denoting the number of edges and vertices of the core of $H[S]$, respectively, Property~$D_{\delta}$ of the core of $H$ implies that 
\begin{equation}\label{eq:EdgesVertices}
(1-\delta)v_{\CP(S)}  >e_{\CP (S)}\geq \gamma |V|/2.
\end{equation}
Since additionally $t + e_{\C(S)} \le |V|$ we infer that
\begin{equation*} \label{eq:density3}
{e_S \over v_S} \leq {t + e_{\CP(S)} \over t + v_{\CP(S)}}\leq  {t + (1-\delta)v_{\CP(S)} \over t + v_{\CP(S)}} = 
1 - \delta {v_{\CP(S)} \over t + v_{\CP(S)}} \stackrel{(\ref{eq:EdgesVertices})}{\leq} 
1 -\delta {\gamma |V| \over 2(1-\delta)|V|} = 1- \delta {\gamma \over 2(1-\delta)}.
\end{equation*}
Together with (\ref{eq:density1})--(\ref{eq:density2}) the above inequality determines the value of $\zeta$. 
Taking 
$$\zeta : = \min \left\{ \frac12, {\gamma \over 2(1-\delta)},  {(k-2 -1/2) \gamma\over 4} \right\}$$ suffices.
\end{proof}
The main ingredient in our proof is statement about the subgraphs of the core itself.
\begin{theorem}  \label{thm:densitycore}
Let $\eps > 0$ and suppose that $m=\lfloor (1-\eps) c_k^*n \rfloor$, where $c_k^*$ is given in Theorem~\ref{thm:main}. Then, for sufficiently small $\eps$,  the core of  
$H_{n,m,k}^*$ has Property $D_{\eps^3/2}$ with probability $1-o(1)$.  
\end{theorem}
The above theorem together with Proposition~\ref{prop:expand} and Lemma~\ref{lem:transfer} yield Theorem~\ref{thm:density}. 
In the remainder of this section we prove Theorem~\ref{thm:densitycore}. 

\subsection{Models of Random Hypergraphs}\label{ssec:models}

Theorem~\ref{thm:densitycore} is stated for the $H_{n,m,k}^*$ model, where multiple edges are allowed, and also each edge can contain a vertex more than once. We start by arguing that it suffices to consider a slightly different random graph model. Let $H_{n,m,k}$ denote a random hypergraph that is created by selecting $m$ edges with $k$ distinct vertices in each edge without replacement. Then the following is true.

\begin{proposition} \label{prop:ModelsEquivalence}
Let~$k\geq 3$ and~$\eps > 0$ be sufficiently small. Assume that~$m=\lfloor cn \rfloor$, for some~$c>0$. If~$\prob (H_{n,m,k} \text{ has property } D_{\eps^3})=1-o(1)$, then
$\prob (H_{n,m,k}^* \text{ has property } D_{\eps^3/2})=1-o(1)$ as well.  
\end{proposition}
\begin{proof} 
First of all, recall that Proposition~\ref{prop:expand} implies that~$H_{n,m,k}^*$ has Property~$\mathcal{E}$ with probability~$1-o(1)$. Therefore, sets of size at most~$\gamma n$, for some sufficiently small~$\gamma>0$, do not violate Property~$D_{\eps^3/2}$. So it is sufficient to argue only about sets with at least~$\gamma n$ vertices. 

Let us call an edge in $H_{n,m,k}^*$ \emph{bad} if it either has repeated vertices or if there is another edge that contains exactly the same
vertices. For each of these edges we resample new edges until the resulting hypergraph contains $m$ different edges with $k$ distinct vertices
in each. Note that this process yields $H_{n,m,k}$. A trivial calculation reveals that with probability $1-o(1)$ the random hypergraph 
$H_{n,m,k}^*$ has at most $\log n$ bad edges. 

In the above process, sets with at least~$\gamma n$ vertices may change their number of edges by at most~$\log n$. Thus, for~$n$ large enough, if~$H_{n,m,k}$ has the property~$D_{\eps^3}$ and there are at most~$\log n$ bad edges, then~$H_{n,m,k}^*$ must have property~$D_{\eps^3/2}$. The statement of the proposition follows. 
\end{proof} 
Thus, proving Theorem~\ref{thm:densitycore} for~$H_{n,m,k}$ (where we use~$\eps^3$ instead of~$\eps^3/2$) is sufficient. Our remaining proof strategy is inspired by the ideas in~\cite{fp10,fp10a}, where it was shown that the core of $H_{n,m,k}$ has property $D_0$. Particularly, we develop further and adapt the arguments in that papers in order to show the stronger statement claimed in Theorem~\ref{thm:densitycore}.

We will use the following auxiliary statements about binomial coefficients.
\begin{proposition}\label{prop:binomial}
Let~$H(x) = -x\log x - (1-x)\log(1-x)$ denote the entropy function. Then, for any~$0 < \alpha<1$ and $|\delta| < 1$, as $n \to\infty$
\[
	\binom{n}{\alpha n} = \frac{1 + o(1)}{\sqrt{2\pi \alpha (1-\alpha)n}}e^{nH(\alpha)}
	~~\text{ and }~~
	\binom{n}{\alpha n + \delta n} \le \binom{n}{\alpha n}e^{n |\delta| \log(\max\{1/\alpha, 1/(1-\alpha)\})}.
\]
\end{proposition}
\begin{proof}
The first statement is well-known and follows immediately from Stirling's approximation of the factorial function; we omit the details. To see the second statement, let us first consider the case $\delta \ge 0$. We can assume that~$\alpha + \delta \le 1$, as otherwise the statement holds trivially. Then
\[
\frac{\binom{n}{\alpha n +\delta n}}{\binom{n}{\alpha n}}
=
\prod_{i=1}^{\delta n }\frac{n-\alpha n -i+1}{\alpha n + i}
\le
\left(\frac{(1-\alpha)n}{\alpha n}\right)^{\delta n}
\le
e^{n \delta \log(1 / \alpha)}.
\]
Similarly, for the case $\delta <0$ we obtain
\[
\frac{\binom{n}{\alpha n - |\delta| n}}{\binom{n}{\alpha n}}
=
\prod_{i=0}^{|\delta| n -1}\frac{\alpha n - i}{n-\alpha n + i +1}
\le
\left(\frac{\alpha n}{(1-\alpha)n}\right)^{|\delta| n}
\le
e^{n |\delta| \log(1 / (1-\alpha))}.
\]
\end{proof}

For the sake of convenience we will carry out our calculations in the~$H_{n,p,k}$ model of random~$k$-graphs. 
This is the ``higher-dimensional" analogue of the well-studied~$G_{n,p}$ model, where, given~$n\geq k$ vertices, we include each~$k$-tuple of vertices with probability~$p$, independently of every other~$k$-tuple. 
Standard arguments show that if we adjust~$p$ suitably, then the~$H_{n,p,k}$ is essentially equivalent to~$H_{n,cn,k}$. The following proposition makes this statement precise.
\begin{proposition} \label{prop:ModelsEquiv}
Let ${\mathcal P}$ be any property of hypergraphs, and let $p = ck/\binom{n-1}{k-1}$, where~$c > 0$. Then
\[
	\Pr {H_{n,cn,k} \not\in {\mathcal P}} \le O(\sqrt{n}) \cdot \Pr {H_{n,p,k}  \not\in {\mathcal P}}.
\]
\end{proposition}
\begin{proof}
Let $N = \binom{n}{k}$, and note that $pN = cn$. Hence, $$\Pr{H_{n,p,k} \text{ has } cn \text{ edges}} = \binom{N}{cn}p^{cn}(1-p)^{N - cn} 
\stackrel{(\text{Prop.~\ref{prop:binomial}})}{=}\Theta(n^{-1/2}).$$ The claim then follows from
\[
	\Pr {H_{n,cn,k} \not\in {\mathcal P}}
	= \Pr {H_{n,p,k} \not\in {\mathcal P} ~|~ H_{n,p,k} \text{ has } cn \text{ edges}}
	\le \frac{\Pr {H_{n,p,k} \not\in {\mathcal P}}}{\Pr{H_{n,p,k} \text{ has } cn \text{ edges}}}.
\]
\end{proof}
In order to prove Theorem~\ref{thm:densitycore} it is therefore sufficient to show that the core of $H_{n,p,k}$ has  Property $D_{\eps^3}$ with probability $1 - o(n^{-1/2})$. This is accomplished in the next sections.

\subsection{Working on the Core of $H_{n,p,k}$: the Cloning Model}
\label{ssec:working}

Recall that the core of a hypergraph is the unique maximum subgraph that has minimum degree (at least)~two. Note that studying properties of the core of a random hypergraph directly is a very difficult task, since the restriction on the minimum degree introduces several dependencies among the edges. At this point we introduce the main tool for our analysis, which provides an accurate description of the core of random hypergraphs. The \emph{cloning model} with parameters $(N, D, k)$, where $N\geq 1$ and $D\geq 0$ are random variables taking integral 
values, is defined as follows. We generate a graph in three stages.
\begin{enumerate}
\item[1.] We expose the value of $N$.
\item[2.] We expose the degrees $\mathbf{d} = (d_1, \dots, d_{N})$, where the $d_i$'s are independent identically distributed as $D$.
\item[3.] For each $1 \le v \le N$ we generate~$d_v$ copies, which we call {\em $v$-clones} or simply {\em clones}. Then we choose uniformly at random a matching from all perfect~$k$-matchings on the set of all clones. Note that such a matching may not exist -- in this case we choose a random matching that leaves less than $k$ clones unmatched. Finally, we construct the graph $H_{\mathbf{d}, k}$ by contracting the clones to vertices, i.e., by projecting the clones of $v$ onto $v$ itself for every $1 \le v \le N$.
\end{enumerate}
Note that the last stage in the above procedure is equivalent to the \emph{configuration model} $H_{\mathbf{d}, k}$ for random hypergraphs with degree sequence~$\mathbf{d}= (d_1, \dots, d_n)$. In other words, $H_{\mathbf{d}, k}$ is a random multigraph where the $i$th vertex has degree $d_i$.
~\\

We will consider a special case of the above model. The so-called \emph{Poisson cloning model}~$\widetilde{H}_{n,p,k}$ for~$k$-graphs with~$n$ vertices and parameter~$p \in [0,1]$, which was introduced by Kim~\cite{ar:k06}. There, we choose~$N = n$ with probability 1, and the distribution~$D$ is the Poisson distribution with parameter~$\lambda := p {n-1 \choose k-1}$. Note that here~$D$ is essentially the vertex degree distribution in the binomial random graph~$H_{n,p,k}$, so we would expect that the two models behave similarly. The following statement confirms this, and is implied by Theorem~1.1 in~\cite{ar:k06}.
\begin{theorem}
\label{cor:Contiguity}
Let~$k \ge 2$ and suppose that~$p = \Theta(n^{-k+1})$. Then there is a~$C>0$ such that for any property~$\mathcal P$ of~$k$-graphs
\[
	\Pr{H_{n,p,k} \not \in {\mathcal P}} \le C\Pr{\widetilde{H}_{n,p,k} \not\in \mathcal{P}}^{1/k} + e^{-n}.
\]
\end{theorem}
In other words, in order to prove Theorem~\ref{thm:densitycore} it is sufficient to show that the core of~$\widetilde{H}_{n,p,k}$ has Property~$D_{\eps^3}$ with probability~$1 - o(n^{-k/2})$. Before we do so, let us collect some important properties of the Poisson cloning model.

One big advantage of the Poisson cloning model is that it provides a very precise description of the core. In particular, Theorem~6.2 in~\cite{ar:k06} implies the following statement, where we write ``$x \pm y$'' for the interval of numbers~$(x - y, x+y)$.
\begin{theorem}\label{thm:CoreClone}
Let $\lambda_{k}:= \min_{x>0} \frac{x}{(1-e^{-x})^{k-1}}$, $0<\delta< 1$, and $c$ be such that $ck=p {n-1\choose k-1} > \lambda_{k}$. Moreover, let~$\bar{x}$ be the largest solution of the equation $x=(1-e^{-xck})^{k-1}$, and set $\xi:= \bar{x}ck$. Then the following is true with probability $1-n^{-\omega(1)}$. If $\tilde{N}_{2}$ denotes the number of vertices in the  core of $\widetilde{H}_{n,p,k}$, then
 $$ \tilde{N}_{2}= (1-e^{-\xi} - \xi e^{-\xi})n \pm \delta n.$$
Furthermore, the core itself is distributed like the cloning model $(\tilde{N}_{2}, \, \mathrm{Po}_{\ge 2}(\Lambda_{c,k}), \, k)$, where $\mathrm{Po}_{\ge 2}(\Lambda_{c,k})$ denotes a
Poisson random variable conditioned on being at least $2$ and parameter~$\Lambda_{c,k}$, where $\Lambda_{c,k}=\xi +\beta$, for some $|\beta|\leq \delta$. 
\end{theorem}
We shall say that a random variable is a $2$-truncated Poisson variable, if it is distributed like a Poisson variable, conditioned on being at least 2.
The next statement is taken from~\cite[Corollary 3.4]{fp10a}.
\begin{corollary}
\label{cor:edgesCore}
Let $\delta > 0$. Let $\tilde{N}_2$ and $\tilde{M}_2$ denote the number of vertices and edges in the core of $\widetilde{H}_{n,p,k}$, where $p = c k / \binom{n-1}{k-1}$ and $ck > \lambda_k$, where $\lambda_k$ is defined in Theorem~\ref{thm:CoreClone}. Then, with probability $1 - n^{-\omega(1)}$,
\[
\tilde{N}_2 = (1-e^{-\xi}-\xi e^{-\xi})n \pm \delta n
\quad\text{ and }\quad
\tilde{M}_2 = {\xi (1-e^{-\xi}) \over k(1-e^{-\xi} - \xi e^{-\xi})} \tilde{N}_2 \pm \delta n,
\]
where $\xi = \bar{x}ck$ and $\bar{x}$ is the largest solution of the equation $x = (1-e^{-xck})^{k-1}$.
\end{corollary}
In the following we will collect a few basic properties of the relation of the number of vertices and edges in the core of $\widetilde{H}_{n,p,k}$. We define the functions $$f(x) = \frac{x(1-e^{-x})}{k(1-e^{-x}-x e^{-x})} ~~~\text{and}~~~ g(x) = \frac{x}{k(1-e^{-x})^{k-1}}$$ and recall that $c_k^*$ and $\xi^*$ in Theorem~\ref{thm:main} and also Theorem~\ref{thm:densitycore} are given by the solution of the system
\begin{equation}
\label{eq:fg}
	1 = f(\xi^*)
\qquad\text{and}	\qquad
	c_k^* = g(\xi^*).
\end{equation}
An easy calculation shows that $f(x)$ is an increasing function of $x$ and infinitely differentiable over $\mathbb{R}^+$, and that $g(x)$ has a unique minimum, which is obtained at $x = x_g$. Moreover,~$g(x_g) = \lambda_k/k$, where $\lambda_k$ is defined in Theorem~\ref{thm:CoreClone}. We shall need the following technical claim.
\begin{claim}
$x_g < \xi^*$.
\end{claim}
\begin{proof}
A simple calculation reveals that
\[
	g'(x) = {\frac { 1-{e^{-x}}-(k-1)x{e^{-x}}}{k ( 1- e^{-x} ) ^{k}}}.
\]
Let $x_0 = 2\log(k-1)$. The numerator of $g'(x_0)$ is 
$1- {1 \over (k-1)^2} - {2\log (k-1) \over k-1}, $
which is easily seen to be greater than zero for all $k\geq 3$. Hence $g'(x_0)>0$ and thus $x_g \leq x_0$. 

In the remainder we argue that $\xi^* \ge {k}/2$, which settles the claim with $x_0 < {k}/2$. Note that the monotonicity of $f$ guarantees that it is enough to show that $f({k}/{2}) \le 1$. Using the estimate $e^x \ge 1 + x + x^2/2$, which is valid for all $x \ge 0$ we obtain
\[
	f\left(\frac{k}{2}\right) =  \frac12 \cdot\frac{e^{k/2} - 1}{e^{k/2} - 1 - k/2} = \frac12\left(1 + \frac{k}{2(e^{k/2}-1-k/2)}\right) \le \frac12\left(1 + \frac4k\right).
\]
Note that for $k \ge 4$ this expression is at most $1$, thus concluding the proof in these cases. Finally, if $k = 3$, then numerical calculations imply that $\xi^* > 2.14 > 2\log 2$.
\end{proof}
Let us assume that~$p = c k / \binom{n-1}{k-1}$, where~$c = (1-\eps)c_k^* > \lambda_k/k$, and set~$\xi = \bar{x}ck$, where~$\bar{x}$ is the largest solution of the equation~$x = (1-e^{-xck})^{k-1}$. So,~$\xi$ is the largest solution of~$c = g(\xi)$, implying with the above claim that~$\xi < \xi^*$. Therefore, we have~$f(\xi) < 1$, and Corollary~\ref{cor:edgesCore} guarantees that with probability~$1 - n^{-\omega(1)}$ the density of the core of~$H_{n,p,k}$ is~less than $1$. This argument can be extended to obtain the following finer statement.
\begin{corollary}
\label{cor:relNM}
Let $\delta > 0$ be sufficiently small and choose $c > \lambda_k/k$ such that the largest solution $\xi$ of the equation $c = g(\xi)$ satisfies $\xi = \xi^* - \delta$, where $\xi^*$ is as in Theorem~\ref{thm:main}. Then there is an $\eps > 0$ such that $\eps = \Theta(\delta)$ and $c = (1-\eps)c_k^*$. Moreover, there is a constant $e_k > 0$ such with probability $1 - n^{-\omega(1)}$
\[
	\tilde{M}_2 = \tilde{N}_2(1 - e_k\delta + \Theta(\delta^2)).
\]
\end{corollary}
\begin{proof}
We first show that there is an $\eps>0$ with the claimed properties. Note that $\xi$ is defined through the equation $c = g(\xi)$ and $\xi^*$
through $c_k^* = g(\xi^*)$. Let $x_g$ be the minimizer of $g$, i.e., $g(x_g) = \lambda_k/k$, and note that whenever $x > x_g$ we have  
$g'(x) > 0$. Applying Taylor's Theorem and using~\eqref{eq:fg} we infer that there is a $\mu \in [\xi, \xi^*]$ such that
\[
	c = g(\xi) = g(\xi^*) + g'(\mu)(\xi - \xi^*) = c_k^* - g'(\mu) \delta.
\]
So $c = (1-\eps)c_k^*$, where $\eps = \frac{g'(\mu)}{c_k^*}\delta$, and note that $g'(\mu)$ remains bounded for $\mu \in [\xi, \xi^*]$,
whenever $\delta$ is sufficiently small.

To see the claim for $\tilde{M}_2$, note that Corollary~\ref{cor:edgesCore} (where we use $\delta^2$ for $\delta$) guarantees that with probability
$1- n^{-\omega(1)}$ we may assume that $\tilde{M}_2 = (f(\xi) \pm \delta^2)\tilde{N}_2$. Moreover, Taylor's Theorem, this time applied to~$f$, implies that 
\[
	f(\xi) = f(\xi^*) + f'(\xi^*)(\xi - \xi^*) + O((\xi - \xi^*)^2)= 1 - f'(\xi^*)\delta + O(\delta^2),
\]
thus concluding the proof with $e_k = {f'(\xi^*)}$ and the fact that $f$ is increasing.
\end{proof}
We immediately obtain the following corollary.
\begin{corollary}
Let $k\ge3$. Let $\eps >0$ be sufficiently small and suppose that $p = (1-\eps)c_k^* k / \binom{n-1}{k-1}$. Then, with probability $1 - n^{-\omega(1)}$
\[
\tilde{M}_2 \le (1-\eps^2)\tilde{N}_2.
\]
\end{corollary}

\subsection{Subgraphs of the 2-Core}
\label{ssec:subgraphs2core}

In order to obtain a sufficient bound for the probability that the core of~$\Pcl$ has property~$D_{\eps^3}$ we will exploit the following statement. Recall that for a hypergraph $H = H(V,E)$ and $S \subseteq V$ we denote by $H[S]$ the subgraph induced by $S$, and we abbreviate $e_S = |E_H(S)|$ for the number of edges in $H[S]$.
\begin{proposition}
\label{prop:Umax1dense}
Let $\gamma > 0$ and let~$H=H(V,E)$ be a $k$-graph such that $|E| < \lfloor(1-\gamma)|V|\rfloor$. Moreover, let~$U$ be an inclusion maximal subset of~$V$ such that $e_U \ge (1-\gamma)|U|$. Then~$e_U = \lceil(1-\gamma)|U|\rceil$ and all edges~$e \in E$ satisfy~$|e \cap U| \neq k-1$.
\end{proposition}
\begin{proof}
If~$e_U>\lceil(1-\gamma)|U|\rceil$, then $e_U \ge (1-\gamma)|U| + 1$. Let~$U' = U \cup \{v\}$, where~$v$ is any vertex in~$V \setminus U$. Note that such a vertex exists, as $U \neq V$. Moreover, denote by~$d$ the degree of~$v$ in~$U$, i.e., the number of edges in~$H$ that contain~$v$ and all other vertices only from~$U$. Then
\[
	e_{U'} = e_U + d \ge e_U \ge (1-\gamma)|U| + 1.
\]
Note that $|U'| = |U| + 1$. Hence, the above inequality implies that
\[
	e_{U'} \ge (1-\gamma)(|U| + 1) - (1 - \gamma) + 1 \ge (1 - \gamma)|U'|,
\]
which contradicts the maximality of~$U$. Similarly, if there was an edge~$e$ such that~$|e \cap U| = k-1$, then we could construct a larger subset of~$V_H$ that also satisfies the density requirement by adding the vertex in~$e \setminus U$ to~$U$.
\end{proof}
Let $\gamma >0$. The following lemma bounds the probability that a given subset of the vertices of the core is maximal and has density $\ge 1- \gamma$, \emph{assuming that the degree sequence has been exposed}. That is, the randomness is that of the 3rd stage of the exposure process in the Poisson cloning model. A similar statement was shown in~\cite{fp10} for the special case $\gamma = 0$.
\begin{lemma} \label{lem:BU}
Let~$k\ge 2$,~$\mathbf{d} = (d_1, \dots, d_N)$ be a degree sequence and~$U \subseteq \{1, \dots, N\}$ such that 
$|U| = \lfloor \beta N \rfloor$, where $1/2 < \beta \le 1$. Moreover, set~$M = k^{-1}\sum_{i=1}^N d_i$ and~$q = (kM)^{-1} \sum_{i\in U} d_i$. Let $0 < \gamma < 1/4$ and assume that~$3N/4 < M < (1-\gamma)N$. If~$\mathcal{B}(\beta, q;\, \gamma)$ denotes the event that~$U$ is an inclusion maximal set of~$H_{\mathbf{d}, k}$ such that~$e_U \ge (1-\gamma)|U|$, then
\[
	\mathbb{P}_{\mathbf{d},k}(\mathcal{B}(\beta, q;\, \gamma)) \le \max\left\{1, \binom{M}{\lfloor \beta N \rfloor}\right\} \cdot (2^k-k-1)^{M- \beta N } \cdot e^{-kM\cdot H(q)} \cdot e^{O(\gamma  \log(1/\gamma) N)},
\]
where $H(x) = -x\ln x - (1-x)\ln(1-x)$ denotes the entropy function, and $\mathbb{P}_{\mathbf{d},k}$ denotes the probability measure on the space of Stage 3, given the outcomes of the first two stages.
\end{lemma}
\begin{proof}
The graph~$H_{\mathbf{d}, k}$ is obtained by creating~$d_i$ clones for each~$1\le i\le N$ and, thereafter, by choosing uniformly at random a perfect~$k$-matching on this set of clones. Note that this is the same as throwing~$kM$ balls into~$M$ bins, with the condition that every obtains exactly~$k$ balls. We use this analogy to prove the claim as follows. Assume that we color the $kqM$ clones of the vertices in~$U$ red, and the remaining $k(1-q)M$ clones blue. So, by applying Proposition~\ref{prop:Umax1dense} we are interested in the probability of the event that there are exactly~$\lceil(1-\gamma)|U|\rceil$ bins with~$k$ red balls and no bin that contains exactly one blue ball.

We estimate the probability for this event as follows. We start by putting into each bin $k$ \emph{black} balls, labeled with the numbers 
$1, \dots, k$. Let~$\mathcal{K} = \{1, \dots, k\}$, and let~$X_1, \dots, X_M$ be independent random sets such that for $1 \le i\le M$
\[
	\forall \mathcal{K}' \subseteq \mathcal{K} ~:~ \Pr{X_i = \mathcal{K}'} = q^{|\mathcal{K}'|}(1-q)^{k - |\mathcal{K}'|}.
\]
Note that $|X_i|$ is distributed like $\Bin(k, q)$. We then recolor the balls in the $i$th bin that are in $X_i$ with red, and all others with blue. 
We infer that the total number of red balls is $X = \sum_{i=1}^M |X_i|$. Set 
\[
	Z = \Pr{X = kqM}. 
\]
Note that $\E{X} = kqM$, and that $X$ is distributed like $\Bin(kM, q)$. By applying Proposition~\ref{prop:binomial} we infer that
\[
	Z = \Pr{X = \E{X}} = (1 + o(1))(2\pi q(1-q)kM)^{-1/2}.
\]
Let $R_j$ be the number of $X_i$'s that contain $j$ elements, and set 
\[
	P = \Pr{X = kqM \wedge R_k = \lceil(1 - \gamma)|U|\rceil \wedge R_{k-1} = 0}. 
\]
Let $e_U = \lceil(1 - \gamma)|U|\rceil$. By using the above notation we may estimate
\begin{equation}
\label{eq:probBU}
	\Pr{\mathcal{B}(\beta, q;\, \gamma)} = \frac{P}{Z} \le \sqrt{2 M} \cdot \Pr{X = kqM \wedge R_k = e_U \wedge R_{k-1} = 0}.
\end{equation}
Let $p_j = \Pr{|X_i| = j} = \binom{k}{j}q^j(1-q)^{k-j}$. Moreover, define the set of integer sequences
\[
	\mathcal{A} = \left\{(b_0, \dots, b_{k-2}) \in \mathbb{N}^{k-1} ~:~ \sum_{j=0}^{k-2} b_j = M - e_U \textrm{ and } \sum_{j=0}^{k-2} j b_j = kqM - ke_U\right\}.
\]
Then
\begin{equation*}
\begin{split}
	& P =  \sum_{(b_0, \dots, b_{k-2})\in \mathcal{A}} \binom{M}{b_0, \dots, b_{k-2}, 0, e_U} \cdot 
\left( \prod_{j=0}^{k-2}p_j^{b_j} \right) \cdot p_k^{e_U}.
\end{split}
\end{equation*}
Observe that the summand can be rewritten as
\[
	\binom{M}{e_U} q^{kqM}(1-q)^{k(1-q)M} \cdot \binom{M-e_U}{b_0, \dots, b_{k-2}} \prod_{j=0}^{k-2} \binom{k}{j}^{b_j}.
\]
By applying the multinomial theorem we obtain the bound
\[
\sum_{(b_0, \dots, b_{k-2})\in \mathcal{A}}\binom{M-e_U}{b_0, \dots, b_{k-2}} \prod_{j=0}^{k-2} \binom{k}{j}^{b_j} \le (2^k-1-k)^{M-e_U}.
\]
Thus, from~\eqref{eq:probBU} we infer that for large $M$
\begin{equation*}
\label{eq:prelim}
	\Pr{\mathcal{B}(\beta, q;\, \gamma)}
	\le
	2\sqrt{M}\binom{M}{e_U} q^{kqM}(1-q)^{k(1-q)M} (2^k-k-1)^{M-e_U}.
\end{equation*}
The proof is completed by estimating $\binom{M}{e_U}$. More specifically, assume first that $|U| \ge (1-\gamma)M$. Then $\gamma < 1/4$ guarantees that
\[
	\binom{M}{e_U} \le \binom{M}{2\gamma M} \le \left(\frac{eM}{2\gamma M}\right)^{2\gamma M} = e^{O(\gamma \log(1/\gamma) \cdot N )}.
\]
Otherwise, let us write $|U| = \lfloor \beta N \rfloor = \eta M$, for some appropriate $\eta \le 1-\gamma$. Note that $\beta > 1/2$ and $M < (1-\gamma)N$ guarantee that $\eta > 1/2$. By applying Proposition~\ref{prop:binomial} we obtain
\[
	\binom{M}{e_U} = \binom{M}{(1 - \gamma)\eta M} \le \binom{M}{\eta M} e^{M\gamma\eta\log(\max\{1/\eta, 1/(1-\eta)\})} = \binom{M}{|U|}e^{O(\gamma \log(1/\gamma) N)}.
\]
\end{proof}
With the above lemma at hand we are ready to estimate the number of subsets of vertices in the core of $\Pcl$ that have density at least~$1-\eps^3$. Suppose that the degree sequence of the core~$\C$ is given by~$\mathbf{d} = (d_1, \dots, d_{N_2})$.  
Then, the number of edges in~$\C$ is~$M_2 = k^{-1}\sum_{i=1}^{N_2} d_i$. For~$q, \beta \in [0,1]$ let~$X_{q, \beta}= X_{q, \beta}(\C) = X_{q, \beta}(\mathbf{d})$ denote the number of subsets of~$\C$ with~$\lfloor \beta \tilde{N}_2 \rfloor$ vertices and total degree~$\lfloor q\cdot k \tilde{M}_2 \rfloor$. (We will omit writing ``$\lfloor . \rfloor$'' from now on.) Note that~$X_{q,\beta}$ is a random variable that depends \emph{only} on the outcomes of the first two stages of the exposure of the core. Let also~$Y_{q,\beta}$ denote the number of these sets that are inclusion maximal and have density at least~$1-\eps^3$.

Let~$\delta > 0$. Moreover, let~$p = ck/\binom{n-1}{k-1}$ be such that the largest solution~$\xi$ of the equation~$g(\xi) = c$ satisfies~$\xi = \xi^* - \delta$, where~$g(\xi^*) = c_k^*$. By applying Corollary~\ref{cor:relNM} we infer that there is a $\eps = \Theta(\delta)$ such that $c = (1-\eps)c_k^*$. Moreover, Theorem~\ref{thm:CoreClone} (where we use $\delta^3$ for $\delta$) guarantees that with probability~$1-n^{-\omega(1)}$
\[
\tilde{N}_2 = n(1-e^{-\xi} - \xi e^{-\xi}) \pm \delta^3 n ~~ \mbox{and} ~~ \Lambda_{c,k} = \xi \pm \delta^3, \qquad \text{where } \xi = \xi^* - \delta.
\]
Set
\[
	n_2 = (1 - e^{-\xi} - \xi e^{-\xi}) n ~\text{ and }~ m_2 = \frac{\xi(1 - e^{-\xi})}{k(1 - e^{-\xi} - \xi e^{-\xi})}n_2
\]
and let $\cal A$ be the event
\begin{equation}
\label{eq:N2M2}
	{\cal A} ~:~ \tilde{N}_2 = n_2 \pm \delta^3 n ~\text{ and }~ \tilde{M}_2 = m_2 \pm \delta^3 n.
\end{equation}
Corollary~\ref{cor:edgesCore} implies that $\Pr {\cal A} = 1-n^{-\omega(1)}$. Moreover, Corollary~\ref{cor:relNM} guarantees the existence of a $e_k > 0$ such that
\begin{equation} \label{eq:n2m2}
 m_2 = (1 - e_k\delta + \Theta(\delta^2))n_2.
\end{equation}
We shall assume all the above facts in the remainder. We are ready to prove the main result of this section, which deals with sets with more than $0.7N_2$ vertices. Smaller sets are treated at the end of this section.
\begin{lemma}
\label{lem:coresubsetslarge}
With the notation above, let $\beta \in [0.7,1 - e_k\delta/2]$ and let $\beta \leq q \leq 1 - {2(1-\beta)/k}$. Then, for sufficiently small $\delta > 0$
\[
\Pr { Y_{q,\beta} >0} = n^{-\omega(1)}.
\]
Moreover, when $q < \beta$ or $q > 1 - 2(1-\beta)/k$, the above probability is 0. 
\end{lemma}
\begin{proof}
The proof follows the arguments in \cite{fp10a}, see Lemma 4.5 -- Claim 4.11 there. We describe here all necessary steps and refer each time we need a statement from~\cite{fp10a} at the appropriate position in that paper. Firstly, suppose that we have exposed the degree sequence~$\mathbf d$ of the core. Markov's inequality implies that
\begin{equation}
\label{eq:tmptmp}
	\Pr { Y_{q,\beta} >0 ~|~ \mathbf{d}} \le X_{q,\beta}(\mathbf{d}) \mathbb{P}_{\mathbf{d}, k}({\cal B}(\beta, q;\, \eps^3)),
\end{equation}
where ${\cal B}(\beta, q;\, \eps^3)$ is as in Lemma~\ref{lem:BU}. 
Note that for sufficiently small $\delta > 0$, ~\eqref{eq:n2m2}, the upper bound on $\beta$, and Proposition~\ref{prop:binomial} imply 
\[
	\binom{m_2 \pm \delta^3n}{\beta(n_2 \pm \delta^3n)}
	\le \binom{n_2}{\beta(n_2 \pm \delta^3n)}
	\le \binom{n_2}{\beta n_2} \cdot e^{O(\delta^3 \log(1/\delta) n_2)}.
\]
By conditioning on~$\cal A$, taking expectations on both sides of~\eqref{eq:tmptmp}, and applying Lemma~\ref{lem:BU} (see Lemma 4.5 in~\cite{fp10}) we infer that for $\delta >0$ sufficiently small
\begin{equation}\label{eq:yqbfirst}
	\Pr { Y_{q,\beta} >0} \le \E{X_{q,\beta} ~|~ \cal A} \cdot \binom{n_2}{\beta n_2} \cdot (2^k-k-1)^{m_2- \beta n_2 } \cdot e^{-km_2\cdot H(q) + O(\delta^2n_2)} + \Pr{\overline{\cal A}}.
\end{equation}
The expectation of $X_{q,\beta}$, conditioned on the event $\cal A$, is determined by calculating the probability that a specific set with $\beta \tilde{N}_2$ vertices  has total degree $qk \tilde{M}_2$. This task was performed in~\cite{fp10a}, see Lemma 4.8 there. It follows that
\[
	\E{X_{q,\beta} ~|~ \cal A} = \exp\left(n_2H(\beta) - n_2(1-\beta)I\left(\frac{k(1-q)}{1-\beta}\right)(1+o(1)) + O(n_2\delta^2)\right),
\]
where 
\begin{equation*}
I(z) = 
\begin{cases} z\left(\ln T_z - \ln \xi \right) - \ln \left(e^{T_z} - T_z -1 \right) + \ln \left(e^{\xi}-\xi -1 \right), & \text{ if $z>2$} \\
\ln 2 - 2\ln \xi + \ln (e^{\xi} - \xi -1),& \text{ if $z=2$}
\end{cases},
\end{equation*}
and $T_z$ is the unique solution of $z={T_z (1-e^{-T_z})\over 1-e^{-T_z} - T_z e^{-T_z}}$. Let
$$f(\beta, q) := 
2~H(\beta) + (1-\beta)\ln (2^k - k -1) -k  H\left(q\right) - (1-\beta) I\left({k(1-q) \over 1-\beta}\right).$$
By using~\eqref{eq:yqbfirst} and~\eqref{eq:n2m2} we infer that
\[
	\Pr { Y_{q,\beta} >0  } \le \exp\left\{n_2\left(f(\beta, q) + e_k\delta\big(kH\left(q\right)- \ln(2^k-k-1)\big) + O(\delta^2)\right)\right\} + n^{-\omega(1)}.
\]
In~\cite{fp10a} the following was shown, see Claim 4.11 there.
\begin{claim}
\label{cl:techn}
There exists a $C > 0$ such that for any small enough $\nu > 0$ the following is true. Let $0.7 \le \beta \le 1-\nu$ and $\beta \leq q \leq 1 - {2(1-\beta)/k}$. Then
\[
	f(\beta, q) \le -C\nu + O(\delta^2).
\]
\end{claim}
We distinguish between the following cases. First, note that if $0.7 \le \beta \le 1 - \sqrt{\delta}$, then the above claim yields for sufficiently small~$\delta > 0$
\[
	\Pr { Y_{q,\beta} >0  } \le e^{n_2(-C\sqrt{\delta} + O(\delta))}  + n^{-\omega(1)} =  n^{-\omega(1)}.
\]
Finally, if $1 - \sqrt{\delta} \le \beta \le 1 - e_k\delta/2$, then the above claim implies that there is a $C' > 0$ such that $f(\beta, q) < C'\delta^2$. Moreover, by the monotonicity of the entropy function and $q \ge \beta$ we have for sufficiently small~$\delta > 0$
\[
	kH\left(q\right) - \ln(2^k-k-1)
	\le kH(0.99) - \ln(2^k-k-1).
\]
A simple calculation and the fact $H(0.99) < 0.06$ show that the above expression is negative for all $k \ge 3$. \end{proof}
This completes the proof of Theorem~\ref{thm:densitycore} for the case $0.7 \le \beta \le 1 - e_k\delta/2$. 
Now if $\beta \geq 1-e_k\delta/2$, then~\eqref{eq:N2M2} together with~\eqref{eq:n2m2} imply that 
for small~$\delta$ all larger subsets have density smaller than~$1-\eps^3$.

In order to cover the remaining cases for $\beta$ we use straightforward first moment arguments. For technical reasons we state our results for the uniform model $H_{n,m,k}$. We begin with the case $k \ge 5$.
\begin{lemma}
\label{lem:kge5spimple}
Let $k \ge 5, c<1$ and $0 < \gamma<0.001$. Then $H_{n,cn,k}$ contains no subset with less than~$0.7n$ vertices and density at least $1-\gamma$ with probability at least $1-n^{-(1-\gamma)k^2+2k+1}$.
\end{lemma}
\begin{proof}
The probability that an edge of $H_{n,cn,k}$ is contained completely in a subset $U$ of the vertex set is $\binom{|U|}{k}/\binom{n}{k} \le (\frac{|U|}{n})^k$. Let $\frac{k}n\le u \le 0.7$. Then the probability that there is a set with $un$ vertices and density at least $1-\gamma$ is at most
\[
\begin{split}
	\binom{n}{un} \cdot \binom{cn}{(1-\gamma)un} u^{k\cdot (1-\gamma)un}
	\le e^{n(H(u) + H((1-\gamma)u)+ (1-\gamma)ku\ln u)},
\end{split}
\]
where $H(x) = -x\ln x - (1-x)\ln x$ denotes the entropy function. It can easily be seen that the exponent has a unique minimum with respect to $u$ in $[0, 0.7]$, implying that it is maximized either at $u = {k}/n$ or at $u = 0.7$. Note that
\[
	H(0.7) + H((1-\gamma)0.7) + (1-\gamma)k \, 0.7\ln(0.7)
	\le H(0.7) + H((1-\gamma)0.7) + (1-\gamma)5 \cdot 0.7\ln(0.7) \le -0.01
\]
and that
\[
	H\left(\frac{k}{n}\right) + H\left((1-\gamma)\frac{k}n\right) + (1-\gamma)\frac{k^2}{n}\ln\left(\frac{k}n\right)
	=
	-\frac{((1-\gamma)k^2-(2-\gamma)k)\ln n}{n} + O\left(\frac1n\right).
\]
So, the maximum is obtained at $u = k/n$, and for large $n$ we conclude that the probability that there is a subset with at most $0.7n$ vertices and density at least $1 - \gamma$ is at most
\[
	\sum_{k/n \le u \le 0.7} n^{-(1-\gamma)k^2+2k} \le n^{-(1-\gamma)k^2 + 2k  + 1}.
\]
\end{proof}
The cases $k\in\{3,4\}$ need a separate treatment. There we will exploit Proposition~\ref{prop:Umax1dense}.
\begin{lemma} \label{lem:1dense3}
Let $0<\gamma<0.001$. Let $H$ be a $k$-graph, where $k \in \{3,4\}$ and call a set $U \subset V_H$ \emph{bad} if 
\[
	e_U = \lceil(1-\gamma)|U|\rceil ~\text{ and }~ \forall e\in E_H : |e \cap U| \neq k-1.
\]
Then, for any $c\le 0.95$ and sufficiently large $n$
\[
	\Pr{H_{n,cn,3} \text{ contains a bad subset $U$ with $\le n/2$ vertices} } = o(1).
\]
and for any $c\le 0.98$ and sufficiently large $n$
\[
	\Pr{H_{n,cn,4} \text{ contains a bad subset $U$ with $\le 3n/4$ vertices} } = o(1).
\]
\end{lemma}
\noindent
The proof is essentially the same as the proof of Lemma~4.3 in \cite{fp10a}, and thus omitted. 
\begin{proof}[Proof of Theorem~\ref{thm:densitycore}]
First of all, let~$k \ge 5$. By applying Lemma~\ref{lem:kge5spimple} we obtain that with high probability,~$H_{n,m,k}$ does not contain a subset of less than $0.7n$ vertices with density at least~$1-\eps^3$. By Proposition~\ref{prop:ModelsEquivalence} this is also true for~$H^*_{n,m,k}$, and in particular also for the core of~$H^*_{n,m,k}$. Concerning larger subsets of the core of~$H^*_{n,m,k}$, by Proposition~\ref{prop:ModelsEquiv} and Theorem~\ref{cor:Contiguity} it suffices to show that the core~$\C = \C(V_\C, E_\C)$ of~$\widetilde{H}_{n,p,k}$, where~$p = ck/\binom{n-1}{k-1}$ and~$m = cn$, contains no subset with more than~$0.7 |V_\C|$ vertices and density at least~$1-\eps^3$ with probability at least~$1 - o(n^{-k/2})$. This follows from Lemma~\ref{lem:coresubsetslarge}, and the proof is completed for~$k \ge 5$.

The cases~$k = 3,4$ require slightly more work. We begin with~$k = 3$. Lemma~\ref{lem:1dense3} and the fact~$c_3^* < 0.95$ guarantee that with high probability the core of~$H_{n,m,k}$ has no subset~$S$ with~$\le n/2$ vertices such that~$e_S = \lceil(1-\eps^3)|S|\rceil$, and there is no edge that contains precisely two vertices in that set. By Proposition~\ref{prop:ModelsEquivalence} this is also true for~$H^*_{n,m,k}$. In particular, by using Proposition~\ref{prop:Umax1dense} we infer that with high probaility the core of~$H^*_{n,m,k}$ does not contain a inclusion maximal subset with at most~$n/2$ vertices and density at least~$1-\eps^3$.  Concerning larger subsets of the core of~$H^*_{n,m,k}$, again by Proposition~\ref{prop:ModelsEquiv} and Theorem~\ref{cor:Contiguity} it suffices to show that the core~$\C = \C(V_\C, E_\C)$ of~$\widetilde{H}_{n,p,k}$, where~$p = ck/\binom{n-1}{k-1}$ and~$m = cn$, contains no subset with more than~$n/2$ vertices and density at least~$1-\eps^3$ with probability at least~$1 - o(n^{-k/2})$. However, by Corollary~\ref{cor:edgesCore} we know that with probability at least $1 - n^{-\omega(1)}$
\[
|V_\C| = (1-e^{-\xi} - \xi e^{-\xi} \pm O(\eps)) n, \text{ where } \xi = 3c(1-e^{-\xi})^2.
\]
Numerical calculations imply that~$|V_\C| \ge 0.63 n$ for any~$\eps>0$ that is small enough.  So,~$\C$ does not contain any inclusion maximal subset with less than~$n/2 \le |V_\C|/(2\cdot0.63) \le 0.77\tilde{N}_2$ vertices an density at least~$1-\eps^3$. This completes together with Lemma~\ref{lem:coresubsetslarge} the proof for~$k = 3$; the case~$k=4$ follows similarly by using the second part of the conclusion of Lemma~\ref{lem:1dense3}, and the fact that~$c_4^* < 0.98$.
\end{proof}

\section{Spanning properties of $H_{n,m,k}^*$} \label{proof:expansion}
\subsection{Proof of Proposition~\ref{prop:expand}}

The proof is similar to that of Lemma 8 in~\cite{inp:fmm09}, but suitably adjusted to our parameters. For ease of notation we write $t=(k-1)s-\delta$; later we will set $\delta = x_ss$ and $\delta=1$, respectively. The expected number of sets in $H_{n,m,k}^*$ containing $s$ edges that span at most $t$ vertices is bounded from above by 
\begin{equation}
\label{eq:tmp1stmom}
\begin{split}
 {m \choose s} {n \choose t} \left({t^k \over n^k} \right)^{s} 
&\le \left({c_k^* n e \over s} \right)^s \left( {n e \over t} \right)^t \left({t \over n} \right)^{ks} = 
n^{-\delta} e^{ks-\delta} (c_k^*s^{-1})^s t^{s+\delta}\\
&\le n^{-\delta}  e^{ks}(c_k^*s^{-1})^s((k-1)s)^{s} (ks)^{\delta}
= \left(\frac{ks}n\right)^\delta \left( c_k^* (k-1)e^k\right)^s. 
\end{split}
\end{equation}
Let $\xi >0$ be such that $(1+\xi)c_k^* = 1$.
To deduce the first claim we observe that for $\delta=x_ss$ the assumption $x_s = {\log_k((k-1)e^k)  / (\log_k(n/s) - 1)}$ implies that
\[
\left(\frac{ks}n\right)^\delta \left(c_k^* (k-1)e^k\right)^s = \left[ \left(\frac{ks}n\right)^{x_s}  \left(c_k^* (k-1)e^k\right)  \right]^s = (1+\xi)^{-s}.
\]
By taking the sum over all $\log\log n< s \le n/k$ we deduce that the probability that there exists a set of $s$ edges of that spans at most $(k-1-x_s)s$ vertices is 
$O\left((1+\xi)^{-{\log  \log n}}\right)= o(1)$.

The proof of the second part follows along the same lines, except that we use slightly cruder upper bounds. In particular, we bound $c_k^* \leq 1$.   
Setting $\delta=1$, we deduce from \eqref{eq:tmp1stmom} that the expected number of sets with $s\leq \log\log n$ edges that span at most $t=(k-1)s-1$ vertices is at most
\[
\frac{k\log\log n}n \cdot \left(ke^k\right)^{\log\log n} = o(1).
\]

\subsection{Proof of Lemma~\ref{lem:growth}}

We will use the following auxiliary fact.
\begin{proposition}\label{prop:as}
For any constants $a,b> 0$ we have that whenever $D=D(a,b)$ is sufficiently large then 
\[
\prod_{i=1}^j \left(1-\frac{a}{ib+D}\right) \ge  j^{-a/b}\cdot (bD)^{-a/b}\cdot e^{-{a^2}/{bD}}\qquad\text{for all $j\ge 2/b$}.
\]
\end{proposition}
\begin{proof}
Assume that $D\ge 2$ is large enough so that $\frac{a}{b+D}\le 0.5$. As $1-x \geq e^{-x-x^2}$ for $x \le 0.5$ we thus obtain
\begin{equation} \label{eq:exponent}
 \prod_{i=1}^j \left(1 - \frac{a}{ib+D}\right) \geq \exp \left(- \sum_{i=1}^j {\frac{a}{ib+D}} - 
\sum_{i=1}^j \left(\frac{a}{ib+D} \right)^2 \right).
\end{equation}
Further, we have 
\begin{equation*}
\begin{split} 
& \sum_{i=1}^j  \frac{a}{ib+D} \leq a \int_{0}^{j} {1 \over bx+D }~dx =\frac{a}b \cdot \left( \log(bj+D) - \log(D) \right)
\le \frac{a}b \cdot  \log(bj+D) .
\end{split}
\end{equation*}
Similarly,
\begin{equation*}
\begin{split} 
& \sum_{i=1}^j  \left(\frac{a}{ib+D}\right)^2 \leq a^2 \int_{0}^{j} {1 \over (bx+D)^2 }~dx =\frac{a^2}{b} \cdot \left( \frac1{D}- \frac1{bj+D}\right) \le \frac{a^2}{bD}.
\end{split}
\end{equation*}
Now observe that for $j\ge 2/b$ we have $bj + D \le bjD$ and thus $ \log(bj+D) \le \log(bj) + \log(D)$. The substitution of these two bounds into (\ref{eq:exponent}) thus yields
\[
\prod_{i=1}^j \left(1 - \frac{a}{ib+D}\right) \geq j^{-a/b} \cdot (bD)^{-a/b} \cdot e^{-{a^2}/{bD}}.
\]
\end{proof}
Let $H=H(V,E)$ be a $k$-graph on $n$ vertices having Property $\mathcal{E}$. We also fix a vertex $v \in V$ and an orientation $o$ of the edges, and for all $i\geq 0$ we let $s_i$ be the number of vertices that are within $o$-distance at most $i$ from $v$.
Note that $s_i=N_{o,i}(v)$, but we shall be using this symbol throughout this section to avoid an unnecessary notational burden.
If all vertices within $o$-distance at most $i$ from $v$ are occupied, then Property $\mathcal{E}$ implies that
\begin{equation} \label{eq:growth}
s_{i+1} \geq 
\begin{cases} 
(k-1-x_{s_i})s_i, & \mbox{if $s_i >  \log\log n$} \\ 
(k-1)s_i, & \mbox{if $s_i \le \log\log n$}.
 \end{cases}. 
\end{equation}
 
\begin{claim}\label{claim:i0}
Let $i_0 := \min \{ i \; : \; s_i > \log\log n  \}$. Then $i_0 \leq \log_{k-1}\log\log n+1$.
\end{claim} 
\begin{proof}
Observe that for all $2\leq i< i_0$, we have 
$ s_i \geq (k-1)s_{i-1}   \ge \ldots \geq (k-1)^{i-1} s_1 = (k-1)^i$,
and the claim follows.
\end{proof}
%

\begin{claim}\label{claim:lowerbound_st}
Let $t_\eps := \lfloor(1-\eps)\log_{k-1}n\rfloor$. Then there exists a $d_k>0$ such that whenever $\eps>0$ is sufficiently small and
$n$ is sufficiently large we have
\[
s_{t_\eps} \ge n^{1-\eps}\cdot e^{-d_k/\eps}.
\]
\end{claim} 
\begin{proof}
Observe that by (\ref{eq:tlevel}) we have $s_{t_\eps}\leq (k-1)^{t_\eps+1}\le (k-1)n^{1-\eps}$. Let $i_0$ be defined as in the previous claim. Hence, we have for all $i_0 \leq i \leq t_\eps$ for any sufficiently large $n$ that
\[
	x_{s_i} \leq {\log_k ((k-1)e^k) \over \eps\log_k n -\log_k(k-1) -1}
	\leq {2\log ((k-1)e^k) \over \eps\log n }.
\]
Thus, for all such $i$ the first part of (\ref{eq:growth}) yields
$$ s_{i+1} \geq (k-1 - x_{s_i})s_i \geq (k-1) 
\Big( 1 - {2\log ((k-1)e^k) \over (k-1)\eps \log n} \Big) s_i.$$
Set $\phi = \phi(\eps, k ,n) = 1 - {2\log ((k-1)e^k) \over (k-1)\eps \log n}$. By applying the above estimate repeatedly and using Claim~\ref{claim:i0} we obtain
\begin{equation*}
\begin{split}
s_{t_\eps}& \geq (k-1)^{t_\eps-i_0}  \phi^{t_\eps-i_0} s_{i_0}
\geq (k-1)^{t_\eps-\log_{k-1}\log\log n} \phi^{\log_{k-1} n} \cdot \log\log n 
\ge \frac{n^{1-\eps}}{k-1} \cdot \phi^{\log_{k-1} n}.
\end{split}
\end{equation*}
Note that
\[
	\phi^{\log_{k-1}n}
	\ge \exp\left\{-\frac{2\log((k-1)e^k)}{\eps (k-1) \log n} \log_{k-1}n + o(1)\right\}
	= \exp\left\{-\frac{2\log((k-1)e^k)}{\eps (k-1) \log(k-1)} + o(1)\right\}.
\]
Since $\frac{\log((k-1)e^k)}{(k-1) \log(k-1)}>0$ for any $k \ge 3$, the claim follows whenever $\eps$ is sufficiently small.
\end{proof}
 
\begin{claim}\label{claim:upperbound:st}
Let $t\geq i_0$, where $i_0$ is as defined in Claim~\ref{claim:i0}. For every $\eps > 0$ sufficiently small, if $s_{t} \leq \eps n$, then for all $0\leq i \leq t-i_0$, we have 
$$ x_{s_{t-i}} \leq  {\log_k ((k-1)e^k) \over i \log_k(k-1-\gamma) + \log_k(1/\eps)-1} \le \gamma,$$
where $\gamma =  {\log_k ((k-1)e^k) \over \log_k(1/\eps) -1}$. 
\end{claim}
\begin{proof}
We will show the statement by induction on $i$. For $i=0$ this is obtained directly from the definition 
of $x_{s_t}$: 
\begin{equation*} 
\begin{split}
x_{s_t} = {\log_k ((k-1)e^k) \over \log_k(n/s_t)-1}  \leq {\log_k ((k-1)e^k) \over \log_k(1/\eps) -1} =\gamma. 
\end{split}
\end{equation*}
Using~\eqref{eq:growth}, for the induction step we have 
\begin{equation*}
\begin{split}
s_{t-(i+1)} \leq {s_{t-i} \over k-1 - x_{s_{t-(i+1)}}} \leq {s_{t-i} \over k-1-\gamma} \leq {s_t \over (k-1-\gamma)^{i+1}} 
\leq {\eps n \over (k-1-\gamma)^{i+1}}. 
\end{split}
\end{equation*}
Thus the definition of $x_{s_{t-(i+1)}}$ yields
\begin{equation*} 
x_{s_{t-(i+1)}} = {\log_k ((k-1)e^k) \over \log_k (n/s_{t-(i+1)})-1} \leq {\log_k ((k-1)e^k)  \over (i+1) \log_k (k-1-\gamma) 
+ \log_k (1/\eps)-1} < \gamma.
\end{equation*}
\end{proof}
The next claim finishes the proof of Lemma~\ref{lem:growth}.
\begin{claim}\label{claim:lemma}
For every $k\geq 3$ and $\zeta>0$ there exists $\eps_0 = \eps_0(\zeta, k)>0$ such that for all $0 < \eps < \eps_0$ and $n$ sufficiently large the following is true. If all vertices within $o$-distance $T:=\log_{k-1}n + \left( {k+\log(k-1) \over (k-1) \log (k-1)} + \zeta \right) \log_{k-1} \log_{k-1} n$ 
from $v$ are occupied, then $s_T > \eps n$. 
\end{claim}   
\begin{proof} 
We prove the claim by contradiction. Assume that $s_T \le \eps n$. Claim~\ref{claim:upperbound:st} implies that 
\[
x_{s_{T-i}} \le  {\log_k ((k-1)e^k) \over i \log_k(k-1-\gamma) + \log_k(1/\eps)-1} 
\]
for all $0\le i \le T-i_0$ and the first part of (\ref{eq:growth}) thus implies that
\begin{equation} \label{eq:recursion}
\begin{split}
s_{T-j} \leq {\eps n \over \prod_{i=1}^j \left(k-1 - x_{s_{T-i}} \right)} \le {\eps n \over (k-1)^j}~ {1\over \prod_{i=1}^j 
\left(1 - {\log_k ((k-1)e^k) / (k-1) \over i \log_k(k-1-\gamma) + \log_k(1/\eps)-1} \right)}. 
\end{split}
\end{equation}
for all $0\le j \le T-i_0$. We now apply Proposition~\ref{prop:as} for $a = \log_k ((k-1)e^k) / (k-1)$ and $b=\log_k(k-1-\gamma)$. Then Proposition~\ref{prop:as} implies such that whenever $\eps$ is sufficiently small (such that 
$\log_k(1/\eps)-1 \ge D$, where $D=D(a,b)$ is as defined in Proposition~\ref{prop:as}), then 
\begin{equation} \label{eq:exponent1}
 \prod_{i=1}^j \left(1 - {\log_k ((k-1)e^k) / (k-1) \over i \log_k(k-1-\gamma) + \log_k(1/\eps)-1} \right) \geq 
 j^{-{\log_k ((k-1)e^k)  \over (k-1)\log_k(k-1-\gamma)}}\cdot \frac1{C_{\eps,k}},
\end{equation}
where $C_{\eps,k}$ is an appropriately defined constant depending only on $\eps$ and $k$. Then substituting the lower bound from (\ref{eq:exponent1}) into (\ref{eq:recursion}) we obtain that for all $j \leq \log_{k-1} n$ we have
\begin{equation}  \label{eq:recbound} 
s_{T-j} \leq {\eps n \over (k-1)^j}~ (\log_{k-1} n)^{{\log_k ((k-1)e^k) \over (k-1)\log_k(k-1-\gamma)}} C_{\eps,k}.
\end{equation} 

If we now set $R_{\eps,k} := C_{\eps,k}\cdot e^{d_k/\eps}$, where $d_k$ is the constant from Claim~\ref{claim:lowerbound_st}, we deduce that for $j:= \eps \log_{k-1} n + {\log_k ((k-1)e^k)  \over (k-1)\log_k(k-1-\gamma)}\log_{k-1} \log_{k-1} n + \log_{k-1} R_{\eps,k}$ we have
\begin{equation} \label{eq:FinalUpperBound}
s_{T-j} \leq \eps  e^{-{d_k/\eps}} n^{1-\eps}.
\end{equation}

If $\eps$ is small enough, then in turn $\gamma$ is small enough so that for $n$ sufficiently large 
$T-j \ge \lfloor (1-\eps)\log_{k-1}n\rfloor$; this, however, contradicts the lower bound from Claim~\ref{claim:lowerbound_st}. 
\end{proof}

\subsection{Proof of Corollary~\ref{lem:expansion}}
The hypergraph $H'$ still has Property~$\mathcal{E}$ with $n$ instead of 
$N=|V(H')| $. Using this, the proof of Corollary~\ref{lem:expansion} follows exactly along the 
lines of the proof of Lemma~\ref{lem:growth}.

\bibliographystyle{plain}

\begin{thebibliography}{10}

\bibitem{abku99}
Y.~Azar, A.~Broder, A.~Karlin, and E.~Upfal. 
\newblock Balanced allocations. 
\newblock {\em SIAM Journal on Computing}, 29(1):180--200, 1999.

\bibitem{ar:csw07}
J.~A. Cain, P.~Sanders, and N.~Wormald.
\newblock The random graph threshold for $k$-orientiability and a fast
  algorithm for optimal multiple-choice allocation.
\newblock In {\em Proceedings of SODA '07}, pp.\ 469--476, 2007.

\bibitem{ar:ct04}
A.~Coja-Oghlan and A.~Taraz.
\newblock Exact and approximative algorithms for coloring $G(n, p)$.
\newblock {\em Random Strucures \& Algorithms}, 24(3):259--278, 2004.

\bibitem{ar:c04}
C.~Cooper.
\newblock The cores of random hypergraphs with a given degree sequence.
\newblock {\em Random Structures \& Algorithms}, 25(4):353--375, 2004.

\bibitem{ar:dm03}
L.~Devroye and P.~Morin.
\newblock Cuckoo hashing: Further analysis.
\newblock {\em Information Processing Letters}, 86(4):215--219, 2003.

\bibitem{unp:Dietz09} M.~Dietzfelbinger, A.~Goerdt,  M.~Mitzenmacher, A.~Montanari, R.~Pagh, and M.~Rink.
\newblock Tight thresholds for cuckoo hashing via {\sc XORSAT}.
\newblock In \emph{Proceedings of ICALP '10}, pp.\ 213--225, 2010.

\bibitem{ar:dw07}
M.~Dietzfelbinger and C.~Weidling.
\newblock Balanced allocation and dictionaries with tightly packed constant size bins.
\newblock {\em Theoretical Computer Science}, 380(1-2): 47 -- 68, 2007.

\bibitem{ar:dk12}
M.~Drmota and R.~Kutzelnigg.
\newblock A precise analysis of Cuckoo hashing.
\newblock {\em ACM Transactions on Algorithms}, 8(2): Article 11, 2012.


\bibitem{b:e06}
R.~S. Ellis.
\newblock {\em Entropy, large deviations, and statistical mechanics}.
\newblock Classics in Mathematics. Springer-Verlag, Berlin, 2006.
\newblock Reprint of the 1985 original.

\bibitem{ar:fr07}
D.~Fernholz and V.~Ramachandran.
\newblock The $k$-orientability thresholds for {$G_{n, p}$}.
\newblock In {\em Proceedings of SODA '07}, pp.\  459--468, 2007.

\bibitem{inc:fpss03}
D.~Fotakis, R.~Pagh, P.~Sanders, and P.~Spirakis.
\newblock Space efficient hash tables with worst case constant access time.
\newblock In {\em Proceedings of S{TACS} '03}, pp.\ 271--282, 2003.

\bibitem{fp10} N.~Fountoulakis and K.~Panagiotou. 
Orientability of random hypergaphs and the power of multiple choices. In \emph{Proceedings of ICALP '10}, pp.\ 348--359, 2010.

\bibitem{fp10a} N.~Fountoulakis and K.~Panagiotou. 
Sharp load thresholds for cuckoo hashing. \emph{Random Structures \& Algorithms}, 41(3): 306--333, 2012.

\bibitem{inp:fkp11} N.~Fountoulakis, M. Khosla and K.~Panagiotou. 
The multiple-orientability thresholds for random hypergraphs. In \emph{Proceedings of SODA '11}, pp.\ 1222--1236, 2011.

\bibitem{inc:fp09}
A.~Frieze and P.~Melsted.
\newblock Maximum matchings in random bipartite graphs and the space utilization of cuckoo hashtables.
\newblock \emph{Random Structures \& Algorithms}, 41(3): 334--364, 2012.

\bibitem{inp:fmm09}
A. M. Frieze, P. Melsted and M. Mitzenmacher.
\newblock An analysis of random-walk cuckoo hashing. 
\newblock \emph{SIAM Journal on Computing}, 40(2):291--308, 2011.

\bibitem{b:jlr00}
S.~Janson, T.~{\L}uczak, and A.~Ruci\'nski.
\newblock {\em Random graphs}.
\newblock Wiley-Interscience Series in Discrete Mathematics and Optimization.
  Wiley-Interscience, New York, 2000.

\bibitem{ar:k06}
J.~H. Kim.
\newblock Poisson cloning model for random graphs.
\newblock Manuscript, 2006.

\bibitem{inc:m09} M.~Mitzenmacher. Some open questions related to cuckoo hashing. 
 In \emph{Proceedings of ESA '09}, pp. 1--10, 2009. 

\bibitem{mrs00}   
M.~Mitzenmacher, A.W.~Richa, and R.~Sitaraman. 
\newblock The power of two random choices: a survey of techniques and results.
\newblock In: {\em Handbook of Randomized Computing}, pp.\ 255--312, 2000.

\bibitem{ar:m05}
M.~Molloy.
\newblock Cores in random hypergraphs and boolean formulas.
\newblock {\em Random Structures \& Algorithms}, 27(1):124--135, 2005.

\bibitem{inp:pr01}
R.~Pagh and F.~Rodler.
\newblock Cuckoo hashing.
\newblock In {\em Proceedings of ESA '01}, pp.\  121--133, 2001.

\bibitem{ar:psw96}
B.~Pittel, J.~Spencer and N.~Wormald.
\newblock Sudden emergence of a giant k-core in a random graph.
\newblock {\em Journal of Combinatorial Theory, Series B}, 67(1):111--151, 1996.

\end{thebibliography}

\end{document}